\documentclass[journal]{IEEEtran}
\usepackage{nicematrix}

\usepackage[table]{xcolor}
\usepackage{amsmath,amssymb,amsfonts,amsthm}
\usepackage{mathtools}

\newtheorem{Lemma}{Lemma}
\newtheorem{Proposition}{Proposition}
\newtheorem{Remark}{Remark}
\usepackage{tikz}
\usetikzlibrary{arrows.meta,positioning,fit,calc}

\usepackage{graphicx}
\usepackage{epstopdf}
\usepackage{pgfplots}
\pgfplotsset{compat=1.18}

\usepackage{float}
\usepackage{stfloats} 

\ifCLASSOPTIONcompsoc
  \usepackage[caption=false,font=normalsize,labelfont=sf,textfont=sf]{subfig}
\else
  \usepackage[caption=false,font=footnotesize]{subfig}
\fi

\usepackage[ruled,vlined,linesnumbered]{algorithm2e}

\usepackage{booktabs}
\usepackage{tabularx}
\usepackage{multirow}
\usepackage{array}
\usepackage{diagbox}

\usepackage{tikz}
\usetikzlibrary{arrows,shapes,positioning}

\usepackage{enumitem}

\usepackage{bbm}
\usepackage{dsfont}

\usepackage{cite}

\usepackage{xcolor}

\usepackage[hidelinks]{hyperref}

\newcommand{\RN}[1]{\textup{\uppercase\expandafter{\romannumeral#1}}}
\begin{document}

\title{Robust SAC-Enabled UAV-RIS Assisted Secure MISO Systems With Untrusted EH Receivers
}

\author{Hamid~Reza~Hashempour, Le-Nam Tran, \textit{Senior Member, IEEE}, Duy H. N. Nguyen, \textit{Senior Member, IEEE} and 
Hien~Quoc~Ngo, \textit{Fellow, IEEE}
	\thanks{

Hamid~Reza~Hashempour and Hien~Quoc~Ngo are with the Center for Wireless Innovation (CWI), Queen’s University Belfast, BT3 9DT Belfast, U.K., (Email:\{h.hashempoor, hien.ngo\}@qub.ac.uk).

Le-Nam Tran is with the School of Electrical and Electronic Engineering,
University College Dublin, Belfield, Dublin 4, D04 V1W8, Ireland (Email:
nam.tran@ucd.ie)

Duy H. N. Nguyen is with the department of Electrical and Computer Engineering, San
Diego State University, San Diego, CA 92182, USA.
(Email: duy.nguyen@sdsu.edu).
	}}

\markboth{}%
{Shell \MakeLowercase{\textit{et al.}}: Bare Demo of IEEEtran.cls for IEEE Journals}

\setlength{\textfloatsep}{8.5pt plus 2pt minus 2pt}

\maketitle


\begin{abstract}
Secure downlink transmission in UAV-mounted reconfigurable intelligent surface (RIS)-assisted multi-user multiple-input single-output (MISO) networks is challenged by imperfect channel state information (CSI), internal eavesdropping by untrusted energy-harvesting receivers (UEHRs), discrete RIS phase shifts, and nonlinear energy harvesting (EH). This paper jointly optimizes the UAV hovering position, base station (BS) transmit power allocation, and discrete RIS phase shifts to maximize the worst-case secrecy energy efficiency (WCSEE) under bounded CSI uncertainty and nonlinear EH constraints. The resulting problem is a highly non-convex mixed-integer fractional optimization problem with tightly coupled spatial, power, and RIS variables. To address these challenges, we develop two complementary solution frameworks. First, for perfect CSI and continuous RIS phase shifts, we derive an analytical benchmark based on block coordinate descent (BCD) and successive convex approximation (SCA). Second, for the general robust setting with discrete phase quantization and CSI uncertainty, we propose a customized soft actor--critic (SAC) framework that incorporates continuous-to-discrete action projection, twin critics to mitigate value overestimation, and a feasibility-aware reward mechanism. Numerical results demonstrate that the proposed SAC policy achieves up to $28\%$, $16\%$, and $6\%$ higher WCSEE than the SCA benchmark, DDPG, and TD3, respectively, while exhibiting faster convergence and maintaining stable performance across varying user densities, UEHR counts, and CSI uncertainty levels.
\end{abstract}

\begin{IEEEkeywords}
UAV-mounted RIS, physical-layer security, SWIPT, worst-case secrecy energy efficiency, successive convex approximation, soft actor--critic, deep reinforcement learning.
\end{IEEEkeywords}

\IEEEpeerreviewmaketitle

\section{Introduction}\label{intro}
\IEEEPARstart{R}{econfigurable} intelligent surfaces (RISs) and unmanned aerial vehicles (UAVs) have emerged as key technologies for enhancing coverage, spectral efficiency, and energy efficiency in beyond-5G and 6G wireless systems~\cite{Survey,Pan2022RISurvey,Zeng2019UAV,Banafaa2024UAV}. By employing low-cost reflecting elements with programmable phase shifts, RISs can reconfigure the wireless propagation environment~\cite{Wu2021Intelligent}. However, terrestrial RISs are constrained by fixed deployment locations and may suffer from blockages, shadowing, and unfavorable propagation conditions. Mounting an RIS on a hovering UAV provides flexible three-dimensional (3D) deployment and facilitates dominant line-of-sight (LoS) links with terrestrial base stations (BSs) and ground users~\cite{Placement-3D,IRS-CUAVN}.

Meanwhile, simultaneous wireless information and power transfer (SWIPT) offers an effective approach for powering energy-constrained Internet of Things (IoT) devices and battery-free sensors~\cite{Clerckx2019WPT,Krikidis2014SWIPT}.  RIS-assisted
SWIPT designs have been widely studied to enhance energy harvesting performance \cite{Zhang,Mohammadi}.  In practical multi-user SWIPT systems, a separated-receiver architecture is commonly considered~\cite{ZhangHo2013}, where the BS simultaneously transmits dual-purpose radio-frequency (RF) signals to legitimate information receivers (IRs) and energy-harvesting receivers (EHRs). While this architecture improves energy sustainability, the broadcast nature of RF transmission creates inherent physical-layer security (PLS) vulnerabilities~\cite{Poor2017PLS,Ouyang}.

In particular, practical nonlinear energy-harvesting (EH) rectifiers require sufficient incident RF power to overcome diode activation thresholds~\cite{Lu2019}. Consequently, EHRs may be located close to the transmitter or receive high-power signal beams, making them potential internal eavesdroppers when they are not authorized to access the transmitted information~\cite{Hashempour,Luo2024STAR}. Such untrusted EHRs (UEHRs) create a fundamental security--energy tradeoff: the network must deliver sufficient RF power to satisfy their EH requirements while limiting the information leakage they may obtain.

This tradeoff becomes more challenging in UAV-mounted RIS systems. First, since UEHRs are untrusted, their channels cannot be assumed to be accurately acquired through cooperative training; therefore, imperfect channel state information (CSI) with bounded uncertainty is required for robust security design~\cite{Imperfection1,Imperfection2}. Second, practical RIS hardware supports only finite-resolution phase shifts, resulting in discrete reflection coefficients~\cite{Wu2019Discrete}. Third, the aerial channel characteristics depend nonlinearly on the UAV's 3D position~\cite{Singh}. The joint design of secure information transmission, nonlinear EH, discrete RIS reflection, and UAV positioning under worst-case CSI uncertainty therefore leads to a highly non-convex, mixed-integer, and high-dimensional optimization problem. These coupled challenges motivate the development of efficient robust optimization and learning-based approaches for worst-case secrecy energy efficiency (WCSEE) maximization.

\subsection{Related Work}\label{sec:related_work}

The literature most relevant to this work can be broadly categorized into four areas: aerial RIS-assisted communications, PLS for SWIPT networks, robust optimization under imperfect CSI, and deep reinforcement learning (DRL) for wireless resource management.

\subsubsection{Aerial RIS-Assisted Networks}

Integrating RISs with UAV platforms provides flexible 3D deployment and favorable propagation conditions, enabling joint optimization of aerial positioning and passive beamforming~\cite{Placement-3D,IRS-CUAVN,STAR-RIS-NOMA,Singh,Li2020UAVRIS}. Misbah \emph{et al.}~\cite{Placement-3D} jointly optimized RIS phase shifts and UAV 3D positioning for downlink sum-rate maximization using alternating optimization (AO). The authors in~\cite{IRS-CUAVN} investigated RIS-assisted cognitive UAV networks by jointly optimizing the sensing duration, passive beamforming, and the 3D position of the leading UAV to maximize network throughput via AO and successive convex approximation (SCA). UAV-assisted simultaneously transmitting and reflecting RIS (STAR-RIS) architectures were explored in~\cite{STAR-RIS-NOMA} to achieve full-space coverage by jointly optimizing 3D UAV positioning, power allocation, and transmission/reflection beamforming for sum-rate maximization, whereas~\cite{Singh} investigated robust aerial RIS beamforming under imperfect CSI. Despite these advances, existing aerial RIS studies primarily focus on information transmission and do not jointly consider EH requirements and internal eavesdropping in SWIPT networks.

\subsubsection{PLS in SWIPT Systems}
While SWIPT has been extensively investigated under both colocated and separated receiver architectures~\cite{ZhangHo2013,Clerckx2019WPT,Zhang,Mohammadi}, incorporating PLS into SWIPT networks introduces unique design challenges~\cite{Luo2024STAR,Lu2019,Ouyang,Hashempour}. Specifically, in~\cite{Luo2024STAR}, robust active and passive beamforming was developed for secure STAR-RIS-assisted SWIPT with nonlinear EH under imperfect CSI. The work in~\cite{Lu2019} optimized global energy efficiency for secure MISO SWIPT with a nonlinear power-splitting EH model. For untrusted EH receivers,~\cite{Ouyang,Hashempour} investigated secure beamforming designs that suppress information leakage while satisfying EH requirements through artificial noise and passive beamforming. However, these studies mainly consider terrestrial systems with fixed transceiver locations and do not jointly address UAV positioning, discrete RIS phase shifts, nonlinear EH, and bounded CSI uncertainty. The resulting WCSEE problem in UAV-mounted RIS-assisted SWIPT therefore remains largely unexplored.
\subsubsection{Robust Optimization Under CSI Uncertainty}

Imperfect channel estimation and feedback errors can significantly degrade PLS, motivating robust transmission designs based on bounded CSI uncertainty~\cite{Imperfection1,Imperfection2}. Norm-bounded worst-case uncertainty models are widely adopted to provide deterministic security guarantees~\cite{proof,hashempour2,Singh}. Conventional robust designs commonly employ the S-procedure, block coordinate descent (BCD), SCA, and related fractional-programming techniques to obtain tractable approximations~\cite{Kumar,Hien,Dinkelbach}. Although these model-based approaches provide useful structural insights and performance benchmarks, their computational burden increases substantially with the system dimension. Moreover, discrete RIS phase shifts introduce non-differentiability, typically requiring continuous relaxation followed by quantization, which can cause performance loss and complicate feasibility preservation. These limitations become particularly pronounced when UAV positioning, transmit power, RIS configuration, nonlinear EH, and worst-case secrecy constraints must be optimized jointly.

\subsubsection{DRL for Wireless Resource Management}

DRL has been increasingly adopted for complex wireless resource-allocation problems, including RIS configuration, beamforming, UAV control, and secrecy optimization~\cite{Niyato,LuoActiveRIS,ZhangRSMA,RazaqNearField}. In~\cite{Niyato}, DRL was employed for ecrecy energy efficiency (SEE) maximization in RIS-assisted wireless networks. Deterministic policy-gradient methods, including deep deterministic policy gradient (DDPG) and twin delayed DDPG (TD3), have also been applied to optimize RIS configurations, UAV trajectories, and beamforming variables~\cite{LuoActiveRIS,ZhangRSMA}. Soft actor--critic (SAC)~\cite{Haarnoja,Haarnoja2}, based on maximum-entropy reinforcement learning, has been investigated as an alternative that improves exploration and training stability through entropy regularization. However, existing DRL-based approaches have not jointly considered discrete RIS phase shifts, UAV hovering optimization, nonlinear EH, and worst-case CSI uncertainty in secure SWIPT systems.

\begin{table*}[t]
	\caption{Comparison of the Proposed Framework With Existing Studies}
	\label{tab:comparison}
	\centering
	\footnotesize
	\setlength{\tabcolsep}{3.8pt}
	\renewcommand{\arraystretch}{1.22}
	\definecolor{rowgray}{gray}{0.91}
	
	\begin{NiceTabular}{|l|c|c|c|c|c|c|c|c|c|c|c|}
		\hline
		\textbf{Feature / Attribute}
		& \textbf{Ours}
		& \cite{Placement-3D,IRS-CUAVN,Singh}
		& \cite{STAR-RIS-NOMA}
		& \cite{Zhang,Mohammadi}
		& \cite{Luo2024STAR}
		& \cite{Lu2019}
		& \cite{Ouyang}
		& \cite{Hashempour}
		& \cite{LuoActiveRIS}
		& \cite{Niyato}
		& \cite{ZhangRSMA} \\
		\hline
		
		\rowcolor{rowgray}
		UAV Location Optimization
		& \checkmark & \checkmark & \checkmark & & & & & & & & \\
		\hline
		
		RIS / STAR-RIS Phase Optimization
		& \checkmark & \checkmark & \checkmark & \checkmark & \checkmark & & & \checkmark & \checkmark & \checkmark & \checkmark \\
		\hline
		
		\rowcolor{rowgray}
		SWIPT / EH
		& \checkmark & & & \checkmark & \checkmark & \checkmark & \checkmark & \checkmark & & & \\
		\hline
		
		Non-linear EH Model
		& \checkmark & & & \checkmark & \checkmark & \checkmark & \checkmark & & & & \\
		\hline
		
		\rowcolor{rowgray}
		Untrusted Receivers (Eavesdroppers)
		& \checkmark & & & & \checkmark & \checkmark & \checkmark & \checkmark & & & \\
		\hline
		
		Worst-Case / Bounded CSI Uncertainty
		& \checkmark & \checkmark & & & \checkmark & & & \checkmark & & & \\
		\hline
		
		\rowcolor{rowgray}
		Discrete RIS Phase Quantization
		& \checkmark & & & & & & & & & \checkmark & \checkmark \\
		\hline
		
		Model-Based Optimization (BCD / SCA)
		& \checkmark & \checkmark & \checkmark & \checkmark & \checkmark & \checkmark & \checkmark & \checkmark & \checkmark & & \\
		\hline
		
		\rowcolor{rowgray}
		DRL
		& \checkmark & & & & & & & & \checkmark & \checkmark & \checkmark \\
		\hline
		
		SEE / WCSEE
		& \checkmark & & & & & \checkmark & \checkmark & & & \checkmark & \\
		\hline
	\end{NiceTabular}
\end{table*}

\subsection{Motivations and Contributions}\label{sec:contributions}

As summarized in Table~\ref{tab:comparison}, existing studies have investigated UAV communications, RIS beamforming, SWIPT, and PLS largely in isolation or under simplified assumptions. A unified framework that jointly captures UAV positioning, discrete RIS phase configuration, active power control, nonlinear EH, and imperfect CSI for secure SWIPT remains lacking. This motivates the development of a robust and computationally efficient design for UAV-mounted RIS-assisted secure networks.

To this end, we consider a secure downlink multi-user MISO SWIPT system in which a hovering UAV-mounted RIS assists a multi-antenna BS in serving multiple legitimate IRs in the presence of distributed UEHRs. Zero-forcing (ZF) precoding is adopted at the BS to suppress inter-user interference among the legitimate IRs. The UAV hovering position, discrete RIS phase shifts, and BS transmit power allocation are jointly optimized to maximize the WCSEE under nonlinear EH requirements and bounded CSI uncertainty. The main contributions are summarized as follows:

\begin{itemize}
\item We formulate a robust WCSEE maximization problem for a UAV-mounted RIS-assisted MISO SWIPT system with nonlinear EH and norm-bounded CSI uncertainty at untrusted receivers. The formulation jointly accounts for discrete RIS phase shifts, UAV hovering constraints, BS transmit-power limits, and minimum harvested-energy requirements.

\item For the continuous-phase and perfect-CSI setting, we develop a model-based BCD algorithm incorporating SCA and the Charnes--Cooper transformation. This framework provides an analytical benchmark for jointly optimizing BS power allocation, RIS reflection phases, and UAV positioning, with convergence to a stationary solution.

\item For the general robust problem with discrete RIS phases and bounded CSI uncertainty, we develop a customized SAC framework. The proposed design employs continuous-to-discrete action projection for phase quantization, twin critics to alleviate value overestimation, and a feasibility-aware reward mechanism to improve constraint satisfaction without extensive manual penalty tuning.

\item Extensive simulations demonstrate that the proposed SAC framework achieves up to $28\%$ and $16\%$ higher WCSEE than the SCA benchmark and DDPG, respectively, while consistently outperforming TD3. Although SCA becomes comparable to DDPG at high transmit powers or with large RIS sizes, SAC maintains the best performance across all considered scenarios. The results further demonstrate the superior sample efficiency and robustness of SAC under different mini-batch sizes, user densities, UEHR numbers, and CSI uncertainty levels.

\end{itemize}

\textit{Organization}: Section~\ref{Sys_Model} presents the system model, channel representations, and the WCSEE problem formulation. Section~\ref{sec:SCA} details the BCD--SCA optimization benchmark. Section~\ref{SAC-section} develops the proposed SAC-based framework. Section~\ref{Simulat} discusses numerical simulation results, and Section~\ref{conc} concludes the paper.

\textit{Notation}: Bold lowercase and uppercase letters denote vectors and matrices, respectively. The operators $(\cdot)^T$, $(\cdot)^H$, and $\mathrm{Tr}(\cdot)$ denote the transpose, conjugate transpose, and trace of a matrix, respectively. $\|\cdot\|_2$ and $|\cdot|$ denote the Euclidean vector norm and complex modulus, respectively. $\mathfrak{Re}\{\cdot\}$ and $\mathfrak{Im}\{\cdot\}$ extract the real and imaginary components. $\mathbb{E}[\cdot]$ denotes statistical expectation. $\mathcal{CN}(\boldsymbol{\mu}, \mathbf{\Sigma})$ represents a circularly symmetric complex Gaussian distribution with mean $\boldsymbol{\mu}$ and covariance matrix $\mathbf{\Sigma}$. $\mathbf{I}_N$ denotes the $N \times N$ identity matrix, $[\mathbf{x}]_m$ denotes the $m$-th element of vector $\mathbf{x}$, and $\mathrm{diag}\{\mathbf{x}\}$ forms a diagonal matrix from vector $\mathbf{x}$. $\mathbb{R}^{M \times N}$ and $\mathbb{C}^{M \times N}$ denote the spaces of $M \times N$ real and complex matrices, respectively.

\section{System Model and Problem Formulation}\label{Sys_Model}
In the considered network, a multi-antenna BS simultaneously transmits confidential information to legitimate IRs and delivers RF energy to UEHRs through downlink transmission. The IRs are trusted users whose primary objective is to reliably decode their intended confidential information streams. In contrast, the UEHRs are authorized network nodes that harvest RF energy to sustain their operations but do not have the security clearance to decode the confidential information intended for the IRs. Therefore, the UEHRs may act as potential non-colluding internal eavesdroppers, resulting in an inherent trade-off between wireless power transfer and PLS. To mitigate severe terrestrial blockages that prevent direct LoS communication between the BS and the ground receivers, a UAV equipped with a RIS is deployed as an aerial passive relay. By jointly optimizing the RIS phase shifts, the UAV horizontal position, and the BS transmit power, the proposed system aims to enhance the reception of confidential signals at the IRs while satisfying the required EH demands of the UEHRs and limiting information leakage to them.

Consider a downlink wireless communication system consisting of one UAV-assisted RIS, $J$ non-colluding UEHRs, $K$ legitimate users, and one BS. The set of $K$ IRs and $J$ UEHRs are denoted by $\mathcal{K} = \{1, \dots, K\}$ and $\mathcal{J} = \{1, \dots, J\}$, respectively. Due to the presence of high obstacles, there exists no LoS channel between the BS and the IRs. 
The UAV, equipped with a RIS, serves as a passive relay to assist communication between the BS and the users. The RIS is composed of $M$ reflecting elements, and the phase shift of each element can be controlled by the UAV.
The BS is equipped with $N_t$ antennas.

All UEHRs, IRs, and the BS are located on the ground. The horizontal coordinates of user $k \in \mathcal{K}$, UEHR $j \in \mathcal{J}$, and the BS are denoted by $\mathbf{w}^I_k = [x_k^I, y_k^I]^T$, $\mathbf{w}^U_j = [x_j^U, y_j^U]^T$, and $\mathbf{w}_b = [x_b, y_b]^T$, respectively. 
The UAV-mounted RIS flies at a fixed altitude $H$ and hovers to enhance communication. The system model is illustrated in Fig.~\ref{fig1}.

\begin{figure} 
    \centering
    \includegraphics[width=\linewidth]{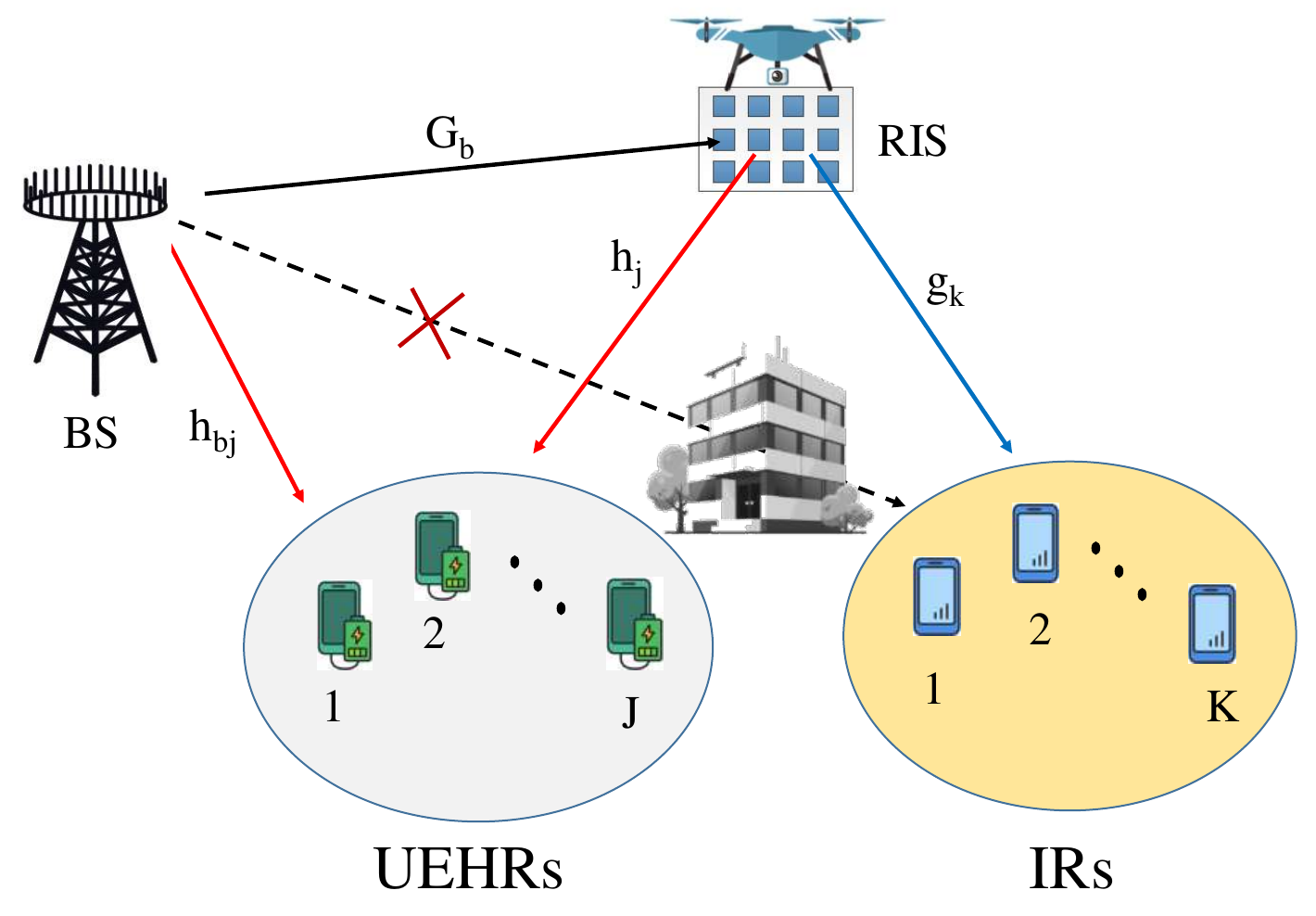}
    \caption{System model of the UAV-assisted RIS-aided secure downlink network.}
    \label{fig1}	
\end{figure}

The UAV (RIS) horizontal coordinate is denoted by $\mathbf{q} = [x_r,y_r]^T$. 
To capture realistic flight restrictions, the UAV is constrained to hover inside a given
rectangular region
\begin{equation}
  \mathcal{L} \triangleq 
  \big\{ \mathbf{q} = [x_r,y_r]^T \,\big|\,
  x_{\min}^r \le x_r \le x_{\max}^r,\, 
  y_{\min}^r \le y_r \le y_{\max}^r \big\}. \label{eq:uav_region}
\end{equation}
Thus, the UAV location must satisfy $\mathbf{q} \in \mathcal{L}$.

\subsection{Channel Model}
Let $d_b$ denote the distance between the BS and the UAV, and
$d_k$ and $d_j$ denote the distances between the UAV and user $k$ and UEHR $j$, respectively, i.e.,
\begin{align} 
    d_b &= \sqrt{\|\mathbf{q} - \mathbf{w}_b\|^2 + H^2},\label{eq:distance1}\\
    d_k &= \sqrt{\|\mathbf{q} - \mathbf{w}^I_k\|^2 + H^2},\quad k\in\mathcal{K},\label{eq:distance2}\\
    d_j &= \sqrt{\|\mathbf{q} - \mathbf{w}^U_j\|^2 + H^2},\quad j\in\mathcal{J}.\label{eq:distance3}
\end{align}
We adopt a distance-dependent path-loss model with exponent $\alpha \ge 2$ and a Rician small-scale fading model for all air-to-ground links. Specifically, the channel between the UAV-mounted RIS and user $k$ is given by $\mathbf{g}_k = \sqrt{\rho_0\, d_k^{-\alpha_k}}\,\tilde{\mathbf{g}}_k$,
and
\begin{align}
\tilde{\mathbf{g}}_k &= 
\sqrt{\frac{K_k}{K_k+1}}\;\mathbf{g}^{\mathrm{LoS}}_{k}
+ \sqrt{\frac{1}{K_k+1}}\;\hat{\mathbf{g}}_{k}, \label{eq:gk_small}
\end{align}
where $\rho_0$ is the path loss at the reference distance of $1$ meter, $\hat{\mathbf{g}}_{k} \sim \mathcal{CN}(\mathbf{0},\mathbf{I}_M)$, and $K_k$ is the Rician $K$-factor. Similarly, the UAV–UEHR $j$ channel is $\mathbf{h}_j = \sqrt{\rho_0\, d_j^{-\alpha_j}}\,
\tilde{\mathbf{h}}_j$, and
\begin{align}
\tilde{\mathbf{h}}_j = \sqrt{\frac{K_j}{K_j+1}}\,\mathbf{h}^{\mathrm{LoS}}_{j}
+ \sqrt{\frac{1}{K_j+1}}\,\hat{\mathbf{h}}_{j},
\end{align}
where $\hat{\mathbf{h}}_{j}\sim \mathcal{CN}(\mathbf{0},\mathbf{I}_M)$ and $K_j$ is the Rician factor.
The LoS components $\mathbf{g}^{\mathrm{LoS}}_{k}$ and $\mathbf{h}^{\mathrm{LoS}}_{j}$ are modeled via the RIS array steering vector.  For simplicity, we adopt a uniform linear array (ULA) model for the RIS. This assumption does not affect the generality of the proposed framework, as it can be readily extended to a uniform planar array (UPA) by incorporating the corresponding three-dimensional (3D) array response.
Specifically, assuming a ULA with inter-element spacing $\delta=\lambda/2$, where $\lambda$ is the carrier wavelength, the steering vector of the RIS towards angle $\varphi$ is
\begin{equation}
    \mathbf{a}_{\mathrm{RIS}}(\varphi)
    = \big[1, e^{-j \frac{2\pi}{\lambda}\delta \sin\varphi}, \ldots,
    e^{-j \frac{2\pi}{\lambda}\delta (M-1)\sin\varphi} \big]^T.
\end{equation}
The array is assumed parallel to the ground, so that the angular dependence of the array response is
captured through the horizontal azimuth angle $\varphi$, while the UAV
altitude $H$ only affects the large-scale path loss via the 3D distance.
Hence, $\mathbf{g}^{\mathrm{LoS}}_{k} = \mathbf{a}_{\mathrm{RIS}}(\varphi_k)$ and 
$\mathbf{h}^{\mathrm{LoS}}_{j} = \mathbf{a}_{\mathrm{RIS}}(\varphi_j)$, where 
$\varphi_k$ and $\varphi_j$ are the angles of departure from the RIS towards user $k$ and UEHR $j$, respectively, determined by $\varphi_k = \arctan\!\left(\frac{y_k - y_r}{x_k - x_r}\right)$ and $\varphi_j = \arctan\!\left(\frac{y_j - y_r}{x_j - x_r}\right)$.

We define the RIS phase shift vector as $\boldsymbol{s} = [s_1, \dots, s_M]^T$, where each element is given by $s_m = e^{j \theta_m}$ for all $m \in \{1, \dots, M\}$. 
The RIS reflection-coefficient matrix is 
\begin{equation}
\boldsymbol{\Theta}
=
\mathrm{diag}\!\left(s_1,\dots,s_M\right)
\in \mathbb{C}^{M\times M},
\label{eq:reflection_matrix}
\end{equation}
In the ideal continuous-phase model, the phase shift satisfies
$\theta_m \in [0,2\pi)$.
In practical implementations with finite-resolution phase shifters,
$\theta_m$ is restricted to a discrete set
\begin{equation}
\theta_m \in \mathcal{F}
\triangleq
\left\{
0,\,
\frac{2\pi}{2^L},\,
\frac{2\pi\cdot 2}{2^L},\,
\dots,\,
\frac{2\pi(2^L-1)}{2^L}
\right\},
\label{eq:discrete_phase}
\end{equation}
where $L$ denotes the number of quantization bits of each RIS element.
The channel matrix between the BS and the RIS is denoted by $\mathbf{G}_b \in \mathbb{C}^{M \times N_t}$, where $[\mathbf{G}_b]_{m,n}$ represents the channel coefficient between the 
$n$th transmit antenna at the BS and the 
$m$th reflecting element of the RIS. It is modeled as
$\mathbf{G}_b = \sqrt{\rho_0\, d_b^{-\alpha_b}}\tilde{\mathbf{G}}_b$,
and 
\begin{equation}
    \tilde{\mathbf{G}}_b=
\left(
\sqrt{\frac{K_b}{K_b+1}}\,\mathbf{G}_b^{\mathrm{LoS}}
+ \sqrt{\frac{1}{K_b+1}}\,\hat{\mathbf{G}}_b
\right)
\end{equation}
where $\hat{\mathbf{G}}_b \sim \mathcal{CN}(\mathbf{0},\mathbf{I}_{M\times N_t})$ and $K_b$ is the Rician $K$-factor. The LoS component is given by
\begin{equation}
\mathbf{G}_b^{\mathrm{LoS}} = \mathbf{a}_{\mathrm{RIS}}(\varphi_{BR})\,\mathbf{a}_{\mathrm{BS}}^H(\varphi_{BR}),
\end{equation}
with $\mathbf{a}_{\mathrm{BS}}(\cdot)$ denoting the BS array response, and $\varphi_{BR}$ the angle of the BS--RIS link.

As a worst-case assumption for secrecy analysis, we also consider a direct link between the BS and each UEHR. 
The corresponding channel vector between the BS and the $j$-th UEHR is denoted by $\mathbf{h}_{bj} \in \mathbb{C}^{N_t \times 1}$ for all $j \in \mathcal{J}$, where the small-scale fading follows a Rayleigh model  with
\begin{equation}
\mathbf{h}_{bj} \sim \mathcal{CN}\big(\mathbf{0}, \rho_0 d_{bj}^{-\alpha_{bj}} \mathbf{I}_{N_t}\big),
\end{equation}
and $d_{bj} = \|\mathbf{w}_b - \mathbf{w}^U_j\|$ is the BS–UEHR distance.

\subsection{Signal Model and ZF Beamforming}
The BS employs linear  ZF precoding to simultaneously serve
$K$ IRs. Let
\[
\mathbf{H}_{\mathrm{c}}
=\big[\,\mathbf{h}_{\mathrm{c},1},\dots,\mathbf{h}_{\mathrm{c},K}\,\big]
\in\mathbb{C}^{N_t\times K}
\]
denote the combined channel matrix of BS--RIS--IR, where
\begin{equation}
\mathbf{h}_{\mathrm{c},k}
=
\mathbf{G}_b^{H}\boldsymbol{\Theta}^{H}\mathbf{g}_k ,
\qquad k\in\mathcal{K}.
\end{equation}
The unnormalized ZF precoding matrix is obtained from the Moore--Penrose
pseudo-inverse as
\begin{equation}
\bar{\mathbf{P}}
=
\mathbf{H}_{\mathrm{c}}
\left(
\mathbf{H}_{\mathrm{c}}^{H}\mathbf{H}_{\mathrm{c}}
\right)^{-1}.
\label{eq:zf_precoder}
\end{equation}
The $k$-th column $\bar{\mathbf{p}}_k$ of $\bar{\mathbf{P}}$
satisfies
$\mathbf{h}_{\mathrm{c},i}^{H}\bar{\mathbf{p}}_k = 0$
for all $i\neq k$.
To incorporate transmit power allocation, each ZF precoder is normalized
and scaled as
\begin{equation}
\mathbf{p}_k
=
\sqrt{\tilde{p}_k}\ \hat{\mathbf{p}}_k,
\qquad k=1,\dots,K,
\label{eq:pk_final}
\end{equation}
where $\tilde{p}_k\ge0$ is the transmit power allocated to user~$k$ and $\hat{\mathbf{p}}_k$ is the corresponding normalized ZF beam direction defined as
\begin{equation}
\hat{\mathbf{p}}_k
\triangleq
\frac{\bar{\mathbf{p}}_k}
{\|\bar{\mathbf{p}}_k\|},
\qquad k=1,\dots,K.
\label{eq:zf_normalized}
\end{equation}
The transmitted signal from the BS is therefore
\begin{equation}
\mathbf{x}
=
\sum_{k=1}^{K}
\mathbf{p}_k\, \tilde s_k,
\label{eq:tx_signal_zf}
\end{equation}
where $\tilde s_k$ is a unit-power information symbol intended for IR~$k$.
Using \eqref{eq:tx_signal_zf}, the received baseband signal at IR~$k$ is
\begin{equation}
y_k
= \mathbf{h}_{\mathrm{c},k}^H\,\mathbf{x} + n_k,
\end{equation}
where $n_k\sim\mathcal{CN}(0,\sigma^2)$ is the additive white Gaussian noise (AWGN). 
Following the ZF property that
$\mathbf{h}_{\mathrm{c},k}^{H}\hat{\mathbf{p}}_l = 0$ for all $\ell\neq k$, the signal-to-interference-plus-noise ratio (SINR) at the legitimate user $k$, obtained as
\begin{equation}
\gamma_k
=
\frac{\big|\mathbf{h}_{\mathrm{c},k}^{H}\mathbf{p}_k\big|^2}{
\sum\limits_{\ell\neq k}\big|\mathbf{h}_{\mathrm{c},k}^{H}\mathbf{p}_\ell\big|^2+\sigma^2 } =\frac{\tilde{p}_k a_k}{\sigma^2}.
\label{eq:gamma_k}
\end{equation}
where $a_k \triangleq \big|\mathbf{h}_{\mathrm{c},k}^{H}\hat{\mathbf{p}}_k\big|^2$.
The received signal at UEHR~$j$ is
\begin{equation}
y_j
=
\big(\mathbf{h}_{bj}^{H} + \mathbf{h}_j^{H}\boldsymbol{\Theta}\mathbf{G}_b\big)
\mathbf{x} + n_j.
\end{equation}
For simplicity of notations we define $\mathbf{u}_j
\triangleq
\mathbf{h}_{bj} + \mathbf{G}_b^{H}\boldsymbol{\Theta}^{H}\mathbf{h}_j$.
Then the eavesdropping SINR for decoding user $k$ is
\begin{equation}
\gamma_{j,k}
=
\frac{\big|\mathbf{u}_j^{H}\mathbf{p}_k\big|^2}{
\sum_{\ell\neq k}\big|\mathbf{u}_j^{H}\mathbf{p}_\ell\big|^2+\sigma^2}
=
\frac{\tilde{p}_k b_{j,k}}{\sum_{\ell\neq k} \tilde{p}_\ell b_{j,\ell}+\sigma^2},
\label{eq:eve_sinr_pow}
\end{equation}
where $b_{j,k}\triangleq |\mathbf{u}_j^{H}\hat{\mathbf{p}}_k|^2$.

\subsection{Worst-Case Secrecy Rate Under CSI Uncertainty}
The BS is assumed to possess perfect CSI (PCSI) for the legitimate cascaded links $\mathbf{h}_{\mathrm{c},k}$, which can be reliably acquired in time-division duplexing (TDD) systems via uplink pilot training and standard cascaded channel estimation protocols (e.g., \cite{ZhengFastCSI,YouRISCSI}). Due to high-elevation line-of-sight propagation, these aerial links exhibit high spatial correlation and change slowly with time, significantly reducing pilot overhead~\cite{GuoCascadedCSI}. In contrast, UEHRs are untrusted nodes that do not engage in reliable channel reporting. Thus, only imperfect estimates $\hat{\mathbf{u}}_j$ are available at the BS~\cite{Hashempour}.

To capture this inaccuracy, we adopt the widely used norm-bounded CSI
uncertainty model \cite{Hashempour,hashempour2,Imperfection1,Imperfection2}. Let
$\hat{\mathbf{u}}_j \in \mathbb{C}^{N_t\times 1}$
denote the estimated cascaded BS--RIS--UEHR channel for UEHR~$j$, and
let $\mathbf{u}_j$ be the true (unknown) channel. The imperfect CSI (ICSI) model is
\begin{equation}
\mathbf{u}_j
=
\hat{\mathbf{u}}_j + \Delta\mathbf{u}_j,
\qquad
\Delta\mathbf{u}_j \in \mathcal{U},
\label{eq:uncertain_channel}
\end{equation}
where the uncertainty set is
\begin{equation}
\mathcal{U}
\triangleq
\Big\{
\Delta\mathbf{u}_j \in \mathbb{C}^{N_t\times 1}
:\;
\|\Delta\mathbf{u}_j\|_2 \le \nu
\Big\}.
\label{eq:uncertainty_set}
\end{equation}
Here, $\nu>0$ defines the radius of the CSI error region and reflects the
UEHR’s infrequent or unreliable feedback. When $\nu=0$, the BS has perfect
knowledge of the UEHR channel. This norm-bounded uncertainty model enables
a worst-case secrecy formulation, ensuring robustness against the strongest
possible UEHR eavesdropping behavior within the uncertainty region. 

Since $\mathbf{u}_j$ is imperfectly known, the BS must ensure robustness against the 
\emph{worst-case eavesdropping capability} within the uncertainty region. 
To account for channel uncertainty, we adopt standard norm-bounded channel uncertainty results.
For any channel model $\mathbf{h}=\hat{\mathbf{h}}+\Delta\mathbf{h}$ with 
$\|\Delta\mathbf{h}\|_2 \leq \nu$ and for any vector 
$\mathbf{p} \in \mathbb{C}^{n}$, the following bounds hold (see \cite{proof,Hashempour})
\begin{subequations}
\label{eq_norm-bounded}
\begin{align}
\max_{\|\Delta\mathbf{h}\|\le\nu} |\mathbf{h}^H \mathbf{p}| 
&= |\hat{\mathbf{h}}^H \mathbf{p}| + \nu \|\mathbf{p}\|_2, \\
\min_{\|\Delta\mathbf{h}\|\le\nu} |\mathbf{h}^H \mathbf{p}| 
&= \Big[ |\hat{\mathbf{h}}^H \mathbf{p}| - \nu \|\mathbf{p}\|_2 \Big]^+,
\end{align}
\end{subequations}
where $[x]^+ = \max(x,0)$.
Applying \eqref{eq_norm-bounded} to the numerator and denominator terms of the eavesdropping SINR in \eqref{eq:eve_sinr_pow} yields
\begin{align}
\max_{\|\Delta\mathbf{u}_j\|\le\nu} \big|\mathbf{u}_j^{H}\mathbf{p}_k\big|^{2} 
&= \Big( \big|\hat{\mathbf{u}}_j^{H}\mathbf{p}_k\big| + \nu\|\mathbf{p}_k\| \Big)^{2}, \\
\min_{\|\Delta\mathbf{u}_j\|\le\nu} \big|\mathbf{u}_j^{H}\mathbf{p}_\ell\big|^{2} 
&=
\Big(\big[|\hat{\mathbf{u}}_j^{H}\mathbf{p}_\ell|-\nu\|\mathbf{p}_\ell\|\big]^+\Big)^{2},
\quad \ell\neq k .
\end{align}

Accordingly, a conservative upper bound on the worst-case eavesdropping SINR for user $k$ at UEHR $j$ is obtained as \footnote{This conservative bound ensures robustness under channel uncertainty at the cost of performance loss.}
\begin{align}
\gamma_{E,k} 
&= \max_{j\in\mathcal{J}}
\left\{
\max_{\|\Delta\mathbf{u}_j\|\le \nu}
\hat{\gamma}_{j,k}(\Delta\mathbf{u}_j)
\right\} \nonumber \\
&\le
\max_{j\in\mathcal{J}}
\frac{
\Big( \big|\hat{\mathbf{u}}_j^{H}\mathbf{p}_k\big| + \nu\|\mathbf{p}_k\| \Big)^{2}
}{
\sum_{\ell\neq k}
\Big(
\big[|\hat{\mathbf{u}}_j^{H}\mathbf{p}_\ell|
-
\nu\|\mathbf{p}_\ell|\big]^+
\Big)^{2}
+\sigma^{2}
} \nonumber \\
&=
\max_{j\in\mathcal{J}}
\frac{\tilde{p}_k
\Big( \sqrt{\hat{b}_{j,k}}+ \nu \Big)^{2}
}{
\sum_{\ell \neq k} \tilde{p}_\ell
\Big(\big[\sqrt{\hat{b}_{j,\ell}}-\nu\big]^+\Big)^{2}
+\sigma^{2}
},
\label{eq:wc_eve_sinr}
\end{align}
where $\hat{b}_{j,k}\triangleq \big|\hat{\mathbf{u}}_j^{H}\hat{\mathbf{p}}_k\big|^2$,
$\forall j\in\mathcal{J},\, k\in\mathcal{K}$.
Given the legitimate-user SINR $\gamma_k$ and the upper bound on the worst-case eavesdropping SINR $\gamma_{E,k}$ in \eqref{eq:wc_eve_sinr}, the worst-case secrecy rate of IR $k$ is expressed as
\begin{equation}
R_k^{\mathrm{sec}}
=
\Big[
\log_{2}(1+\gamma_k)
-
\log_{2}(1+\gamma_{E,k})
\Big]^{+}.
\label{eq:secrecy_rate}
\end{equation}
This formulation guarantees secrecy robustness against all UEHR channel realizations satisfying $\|\Delta\mathbf{u}_j\|\le\nu$.

The total received RF power at UEHR~$j$ is
\begin{equation}
P_j^{\mathrm{EH}}
=
\sum_{k=1}^{K}
\big|\mathbf{u}_j^{H}\mathbf{p}_k\big|^2,
\end{equation}
and the harvested direct current (DC) power follows the logistic nonlinear model
\cite{Boshkovska}
\begin{equation}
\Omega(P_j^{\mathrm{EH}})
=
\frac{b_2}{
k'_1\big(1+\exp(-b_0(P_j^{\mathrm{EH}} - b_1))\big)} - k'_2.
\label{eq:EH_model}
\end{equation}
where $\Omega(\cdot)$ is a logistic function representing the non-linear energy harvesting model, and $k'_1$, $k'_2$, $b_0$, $b_1$, $b_2$ are positive constants.
The required EH power must satisfy the following robust constraint under channel uncertainty
\begin{align}\label{EH-c}
   \sum_{j \in \mathcal{J}} \min_{\|\Delta\mathbf{u}_j\|\le \nu} P_j^{\mathrm{EH}} \geq \Omega^{-1}(E_h),
\end{align}
where $E_h$ denotes the minimum harvested energy requirement at the set of UEHRs. The inverse function $\Omega^{-1}(x)$ is defined as
\begin{equation}
\Omega^{-1}(x) = b_1 - \dfrac{1}{b_0} \ln\left( \frac{b_2}{k'_1(x + k'_2) }- 1  \right). \label{omegainv}
\end{equation}
To obtain a tractable deterministic formulation, \eqref{EH-c} is enforced via the
following sufficient lower-bound constraint
\begin{equation}\label{P1c_closed}
\sum_{j \in \mathcal{J}} \underline{P}_j^{\mathrm{EH}}
\ \ge\ \Omega^{-1}(E_h),
\end{equation}
where
\begin{equation}\label{eq:eh_lower_closed}
\underline{P}_j^{\mathrm{EH}}
\triangleq
\sum_{k=1}^{K}
\left(\left[\,\big|\hat{\mathbf{u}}_j^{H}\mathbf{p}_k\big|
-\nu\|\mathbf{p}_k\|_2\,\right]^+\right)^{2},\qquad \forall j\in\mathcal{J}.
\end{equation}

\subsection{Problem Formulation}\label{Problem-formulation}

To guarantee confidentiality against UEHRs, the secrecy rate of each legitimate IR is defined under worst-case eavesdropping as in \eqref{eq:secrecy_rate}.
To ensure fairness among users, we adopt the minimum secrecy rate
\begin{equation}
R^{\mathrm{sec}}
=
\min_{k\in\mathcal{K}} R_k^{\mathrm{sec}}.
\label{eq:min_sec_rate}
\end{equation}

We aim to maximize the WCSEE,
defined as the ratio between the minimum secrecy rate and the total power
consumption at the BS. The total power includes the radiated transmit power
and the circuit power consumption. The optimization variables include the
ZF power allocation, RIS phase shifts, and the UAV horizontal position.
The resulting problem is formulated as
\begin{subequations}\label{P1}
\begin{align}
\max_{\tilde{\mathbf{p}},\,\boldsymbol{\theta},\,\mathbf{q}}
\quad
&
\frac{
R^{\mathrm{sec}}
}{
\varrho \sum_{k=1}^{K}\tilde{p}_k + P_0
}
\label{P1a}
\\[0.5mm]
\text{s.t.}\quad
&
\sum_{k=1}^{K} \tilde{p}_k \le P_{\max},
\label{P1d}
\\
&
\theta_m \in \mathcal{F},\quad m=1,\dots,M,
\label{P1g}
\\
&
\mathbf{q} \in \mathcal{L},
\label{P1h}
\\&
\eqref{P1c_closed}.
\end{align}
\end{subequations}
where $\varrho$ is the reciprocal of the
drain efficiency of the power amplifier at the BS \cite{Hien}, $\tilde{\mathbf{p}}= \{ \tilde{p}_k \}_{k=1}^{K}$, $\boldsymbol{\theta} = \{ \theta_m \}_{m=1}^{M}$ and $P_0$ represents the circuit power consumption.

\begin{Remark}
Problem~\eqref{P1} is highly nonconvex due to the fractional objective, coupled UAV--RIS channel geometry, worst-case secrecy constraints, and discrete RIS phase shifts. This motivates the use of advanced optimization and learning-based approaches to handle the resulting high-dimensional problem.
\end{Remark}

\section{SCA-Based Solution}
\label{sec:SCA}

To obtain a tractable benchmark, we consider an idealized setting with perfect CSI and continuous RIS phase shifts. Under these assumptions, problem~\eqref{P1} admits a simplified formulation that can be efficiently solved via a BCD framework combined with SCA, and thus reduces to
\begin{subequations}\label{P2}
\begin{align}
\max_{\tilde{\mathbf{p}},\,\boldsymbol{\theta},\,\mathbf{q}}
\quad
&
\frac{
R^{\mathrm{sec}}
}{
\varrho \sum_{k=1}^{K}\tilde{p}_k + P_0
}
\label{P2a}
\\[0.5mm]
\text{s.t.}\quad
&
\sum_{j \in \mathcal{J}} P_j^{\mathrm{EH}} \geq \Omega^{-1}(E_h),
\label{P2c}
\\
&
\theta_m \in [0, 2\pi], \quad \forall m \in \{1, \dots, M\},
\label{P2g}
\\
&
\eqref{P1d},\eqref{P1h}.
\end{align}
\end{subequations}
Problem~\eqref{P2} is addressed using a BCD framework. Specifically, the power allocation, RIS reflection coefficients, and UAV location are optimized in an alternating manner, while the remaining variables are fixed in each subproblem.
\subsection{Power Allocation}
\label{subsec:sca_power}

The RIS configuration $\boldsymbol{\Theta}$
and the UAV location $\mathbf{q}$ are assumed fixed.
Accordingly, the BS computes the ZF beamforming directions from the
effective channel and only optimizes the power loading.
Given fixed $\boldsymbol{\Theta}$ and $\mathbf{q}$, the power allocation
benchmark problem is formulated as
\begin{subequations}\label{Ppow_frac}
\begin{align}
\max_{\tilde{\mathbf{p}}\succeq\mathbf{0}}
\quad
&
\frac{R^{\mathrm{sec}}(\tilde{\mathbf{p}})}{\varrho\sum_{k=1}^{K}\tilde{p}_k+P_0}
\label{Ppow_frac_a}
\\
\text{s.t.}\quad
&
\sum_{j\in\mathcal{J}}\sum_{k=1}^{K} \tilde{p}_k\, b_{j,k}
\ge \Omega^{-1}(E_h),
\label{Ppow_frac_b}
\\
&
\eqref{P1d}.
\label{Ppow_frac_c}
\end{align}
\end{subequations}
In \eqref{Ppow_frac_b}, we have used
\eqref{eq:pk_final}, under which the harvested RF power at UEHR~$j$ is
given by
\begin{equation}
P_j^{\mathrm{EH}}(\tilde{\mathbf{p}})
=
\sum_{k=1}^{K}\tilde{p}_k\, b_{j,k}.
\label{eq:eh_linear_pow}
\end{equation}
As a result, the EH constraint in \eqref{Ppow_frac_b} is
affine with respect to the power allocation vector $\tilde{\mathbf{p}}$.

Problem \eqref{Ppow_frac} is a nonlinear fractional program due to the
fractional objective and the nonconvex secrecy-rate term. Introducing a slack
variable $\zeta$, the secrecy constraint can be equivalently enforced for all
UEHRs as
\begin{align}
\zeta
\le\;
&\log_2\!\Big(\sigma^2+\tilde{p}_k a_k\Big)-\log_2(\sigma^2)
+\log_2\!\Big(\sigma^2+\sum_{\ell\neq k}\tilde{p}_\ell b_{j,\ell}\Big)
\nonumber\\
&-\log_2\!\Big(\sigma^2+\sum_{\ell=1}^{K}\tilde{p}_\ell b_{j,\ell}\Big),
\quad \forall k\in\mathcal{K},\forall j\in\mathcal{J}.
\label{eq:dc_sec}
\end{align}
The last term in \eqref{eq:dc_sec} is upper bounded by its first-order Taylor
expansion at $\tilde{\mathbf{p}}^{(t)}$ as
\begin{align}
\log_2\!\Big(\sigma^2+\sum_{\ell=1}^{K}\tilde{p}_\ell b_{j,\ell}\Big)
\le
\log_2\!\big(S_j^{(t)}\big)
+
\sum_{\ell=1}^{K}\eta_{j,\ell}^{(t)}
\big(\tilde{p}_\ell-\tilde{p}_\ell^{(t)}\big),
\label{eq:taylor_log}
\end{align}
where
$S_j^{(t)}=\sigma^2+\sum_{\ell=1}^{K}\tilde{p}_\ell^{(t)}b_{j,\ell}$ and
$\eta_{j,\ell}^{(t)}={b_{j,\ell}}/{(\ln(2)S_j^{(t)})}$. Substituting
\eqref{eq:taylor_log} into \eqref{eq:dc_sec} gives the concave lower bound
\begin{align}
\zeta \le \widehat{R}_{k,j}^{(t)}(\tilde{\mathbf{p}}),\quad
\forall k\in\mathcal{K},\forall j\in\mathcal{J},
\label{eq:sec_sca}
\end{align}
where
\begin{align}
&\widehat{R}_{k,j}^{(t)}(\tilde{\mathbf{p}})
\triangleq\;
\log_2\!\Big(\sigma^2+\tilde{p}_k a_k\Big)-\log_2(\sigma^2)-\log_2\!\big(S_j^{(t)}\big)
\nonumber\\ &+\log_2\!\Big(\sigma^2+\sum_{\ell\neq k}\tilde{p}_\ell b_{j,\ell}\Big)
-\sum_{\ell=1}^{K}\eta_{j,\ell}^{(t)}
\big(\tilde{p}_\ell-\tilde{p}_\ell^{(t)}\big).
\label{eq:Rhat_pow}
\end{align}

To handle the fractional objective, we apply the Charnes--Cooper transformation
\begin{equation}
\chi =
\frac{1}{\varrho\sum_{k=1}^{K}\tilde{p}_k+P_0},
\qquad
\bar{p}_k=\chi\tilde{p}_k,
\qquad
\bar{\zeta}=\chi\zeta.
\label{eq:cc_transform_pow}
\end{equation}
Thus, at SCA iteration $t$, the power-allocation subproblem \eqref{Ppow_frac} is transformed into
the following convex problem
\begin{subequations}\label{Ppow_cc_sca}
\begin{align}
\max_{\bar{\mathbf{p}}\succeq\mathbf{0},\,\bar{\zeta},\,\chi>0}
\quad
&
\bar{\zeta}
\label{Ppow_cc_sca_a}
\\
\text{s.t.}\quad
&
\varrho\sum_{k=1}^{K}\bar{p}_k+P_0\chi=1,
\label{Ppow_cc_sca_b}
\\
&
\sum_{k=1}^{K}\bar{p}_k \le P_{\max}\chi,
\label{Ppow_cc_sca_c}
\\
&
\sum_{j\in\mathcal{J}}\sum_{k=1}^{K}\bar{p}_k b_{j,k}
\ge \chi\Omega^{-1}(E_h),
\label{Ppow_cc_sca_d}
\\
&
\bar{\zeta}
\le
\chi\widehat{R}_{k,j}^{(t)}
\!\left(\frac{\bar{\mathbf{p}}}{\chi}\right),
\quad \forall k\in\mathcal{K},\forall j\in\mathcal{J}.
\label{Ppow_cc_sca_e}
\end{align}
\end{subequations}
Problem \eqref{Ppow_cc_sca} is convex since
$\chi\widehat{R}_{k,j}^{(t)}(\bar{\mathbf{p}}/\chi)$ is the perspective of the
concave lower bound in \eqref{eq:sec_sca}, and the remaining constraints are
affine. After solving \eqref{Ppow_cc_sca}, the original power allocation is
recovered as $\tilde{p}_k^\star=\bar{p}_k^\star/\chi^\star$. Unlike
Dinkelbach's method \cite{Dinkelbach}, the Charnes--Cooper transformation avoids
an additional fractional-programming loop and requires only one convex solve per
SCA iteration.
This leads to a more compact and computationally efficient power-allocation procedure, as summarized in Algorithm~\ref{Alg:pow_cc_sca}.

\begin{algorithm}[t]
\caption{Charnes--Cooper-SCA Power Allocation}
\label{Alg:pow_cc_sca}
\KwIn{$\tilde{\mathbf{p}}^{(0)}$, tolerance $\varepsilon$.}
\KwOut{$\tilde{\mathbf{p}}^\star$}
Set $t=0$ and $\Delta=\infty$\;
\Repeat{$\Delta\le\varepsilon$}{
Construct $\widehat{R}_{k,j}^{(t)}(\tilde{\mathbf{p}})$ using \eqref{eq:Rhat_pow}\;
Solve \eqref{Ppow_cc_sca} and recover
$\tilde{p}_k^{(t+1)}=\bar{p}_k^\star/\chi^\star$, $\forall k$\;
Set $\Delta=\|\tilde{\mathbf{p}}^{(t+1)}-\tilde{\mathbf{p}}^{(t)}\|_2$ and
$t\leftarrow t+1$\;
}
\Return $\tilde{\mathbf{p}}^{(t)}$.
\end{algorithm}


\subsection{RIS Reflection Phase Optimization}
\label{subsec:sca_ris}

In this subsection, we optimize the RIS reflection coefficients for fixed
ZF beamforming directions, power allocation $\tilde{\mathbf{p}}$,
and UAV location $\mathbf{q}$.
Since the transmit power consumption at the BS is fixed under given
$\tilde{\mathbf{p}}$, the WCSEE
maximization reduces to a worst-case secrecy rate maximization problem.
Accordingly, the RIS phase optimization subproblem can be formulated as
\begin{subequations}\label{Pris}
\begin{align}
\max_{\mathbf{s},\,\zeta}
\quad
& \zeta
\label{Pris_obj}
\\
\text{s.t.}\quad
& R_k^{\mathrm{sec}} \ge \zeta, \quad \forall k \in \mathcal{K},
\label{Pris_sec}
\\
& \eqref{P2c},\eqref{P2g}.
\end{align}
\end{subequations}

For any vectors $\mathbf{v}\in\mathbb{C}^{M}$ and $\mathbf{w}\in\mathbb{C}^{N_t}$,
the cascaded channel term can be rewritten as
$\mathbf{v}^H \boldsymbol{\Theta}\mathbf{G}_b \mathbf{w}
=
\mathbf{t}^H \mathbf{s}$,
where $\mathbf{t} \triangleq (\mathrm{diag}(\mathbf{v}^H)\mathbf{G}_b\mathbf{w})^*$.
Accordingly, we define
\begin{align}
\mathbf{t}_{k,\ell} &\triangleq
\big(\mathrm{diag}(\mathbf{g}_k^H)\mathbf{G}_b\mathbf{p}_{\ell}\big)^{*},
\qquad \forall k\in\mathcal{K},\ \forall \ell\in\mathcal{K},
\\
\mathbf{t}_{j,\ell} &\triangleq
\big(\mathrm{diag}(\mathbf{h}_j^H)\mathbf{G}_b\mathbf{p}_{\ell}\big)^{*},
\qquad \forall j\in\mathcal{J},\ \forall \ell\in\mathcal{K},
\\
c_{j,\ell} &\triangleq \mathbf{h}_{bj}^H\mathbf{p}_{\ell},
\qquad \forall j\in\mathcal{J},\ \forall \ell\in\mathcal{K}.
\end{align}
Then the useful and interference terms become
\begin{align}
\mathbf{h}_{\mathrm{c},k}^H\mathbf{p}_{\ell}
&= \mathbf{g}_k^H\boldsymbol{\Theta}\mathbf{G}_b\mathbf{p}_{\ell}
= \mathbf{t}_{k,\ell}^H\mathbf{s},
\\
\mathbf{u}_j^H\mathbf{p}_{\ell}
&= \mathbf{h}_{bj}^H\mathbf{p}_\ell+\mathbf{h}_j^H\boldsymbol{\Theta}\mathbf{G}_b\mathbf{p}_\ell
= c_{j,\ell}+\mathbf{t}_{j,\ell}^H\mathbf{s}.
\end{align}

Using the definitions above, the received signal powers appearing in
$\gamma_k$ and $\gamma_{j,k}$ can be expressed as quadratic forms in
$\mathbf{s}$.
The secrecy constraint \eqref{Pris_sec} is nonconvex due to the difference
of logarithmic functions.
Following the SCA principle, we introduce first-order approximations of
the nonconvex quadratic terms.
Introduce auxiliary variables $\rho_k$ and $\rho_{E,j,k}$ such that
\begin{subequations}\label{eq:aux_sinrs}
\begin{align}
\rho_k &\le \gamma_k,\qquad \forall k\in\mathcal{K}, \label{eq:aux_sinrs1}
\\
\rho_{E,j,k} &\ge \gamma_{j,k},\qquad \forall j\in\mathcal{J},\ \forall k\in\mathcal{K}.\label{eq:aux_sinrs2}
\end{align}
\end{subequations}
Also introduce $f_{E,k}$ to handle the worst UEHR,
\begin{align}
f_{E,k} &\ge \log_2(1+\rho_{E,j,k}),
\qquad \forall j\in\mathcal{J},\ \forall k\in\mathcal{K},
\label{eq:aux_eve_log}
\end{align}
\begin{Proposition}
\label{prop:ris_zf_affine}
An affine approximation of the secrecy-related constraints in \eqref{Pris_sec},
for all $k\in\mathcal{K}$ and $j\in\mathcal{J}$, is given by
\begin{subequations}\label{eq:ris_zf_affine}
\allowdisplaybreaks
\begin{align}
& \sum_{\ell\neq k} \big|\mathbf{t}_{k\ell}^H \mathbf{s}\big|^2 + \sigma^2
- \Psi^{(t)}\!\big(\mathbf{s},\rho_k;\mathbf{t}_{kk}\big) \le 0,
\quad \forall k\in\mathcal{K},
\label{eq:ris_zf_affine_a}
\\
& \big|c_{j,k} + \mathbf{t}_{j,k}^H \mathbf{s}\big|^2
\le \varTheta^{(t)}\!\big(\xi_{j,k},\rho_{E,j,k}\big),
\quad \forall j\in\mathcal{J},\ \forall k\in\mathcal{K},
\label{eq:ris_zf_affine_b}
\\
& \xi_{j,k}
\le
\sum_{\ell\neq k}\vartheta^{(t)}\!\big(c_{j,\ell},\mathbf{t}_{j,\ell};\mathbf{s}\big)
+\sigma^2,
\quad \forall j\in\mathcal{J},\ \forall k\in\mathcal{K},
\label{eq:ris_zf_affine_c}
\\
& 1 - \Gamma^{(t)}(f_{E,k}) + \rho_{E,j,k} \le 0,
\quad \forall j\in\mathcal{J},\ \forall k\in\mathcal{K},
\label{eq:ris_zf_affine_d}
\\
& \zeta \le \log_2(1+\rho_k) - f_{E,k},
\quad \forall k\in\mathcal{K}.
\label{eq:ris_zf_affine_e}
\end{align}
\end{subequations}
where $\Psi^{(t)}(\mathbf{s},\rho_k;\mathbf{t}_{kk})$ is the first-order Taylor
approximation of $\frac{\big|\mathbf{t}_{kk}^H\mathbf{s}\big|^2}{\rho_k}$,
$\Gamma^{(t)}(x)\triangleq 2^{x^{(t-1)}}\big[1+\ln(2)\,(x-x^{(t-1)})\big]$
is the first-order Taylor approximation of $2^x$, and
\begin{align}
\varTheta^{(t)}(x,y) \triangleq &
 \tfrac{1}{2} \big(x^{(t-1)}+y^{(t-1)}\big)(x+y) \nonumber \\&
-\tfrac{1}{4}\big(x^{(t-1)}+y^{(t-1)}\big)^2-\tfrac{1}{4}(x-y)^2,
\end{align}
is used to obtain an affine lower bound of the bilinear term $xy$.
Moreover, $\vartheta^{(t)}(c;\mathbf{t};\mathbf{s})$ is the first-order Taylor
approximation of $\big|c+\mathbf{t}^H\mathbf{s}\big|^2$ at iteration $t$:
\begin{equation}
\vartheta^{(t)}(c;\mathbf{t};\mathbf{s})
=
2\,\mathfrak{Re}\!\left\{
\big(c+\mathbf{t}^H\mathbf{s}^{(t-1)}\big)^{H}\mathbf{t}^H\mathbf{s}
\right\}
-\big|c+\mathbf{t}^H\mathbf{s}^{(t-1)}\big|^2.
\end{equation}
\end{Proposition}

\begin{proof}
See Appendix~\ref{app:proof_prop_ris_zf}.
\end{proof}
Constraint \eqref{P2c} is nonconvex due to the quadratic terms in $\mathbf{s}$,
and the unit-modulus constraint $|s_m|=1$ is also nonconvex.
In the following, we construct convex surrogate constraints using SCA.
Under fixed $\{\mathbf{p}_k\}$ and $\mathbf{q}$, the harvested RF power at UEHR~$j$ can be written as
\begin{equation}
P_j^{\mathrm{EH}}
=
\sum_{k=1}^{K}\big|c_{j,k}+\mathbf{t}_{j,k}^H\mathbf{s}\big|^2.
\end{equation}
Since $\big|c_{j,k}+\mathbf{t}_{j,k}^H\mathbf{s}\big|^2$ is convex in $\mathbf{s}$,
its first-order Taylor expansion $\vartheta^{(t)}(c_{j,k};\mathbf{t}_{j,k};\mathbf{s})$
provides a global affine lower bound. Hence, an SCA-safe sufficient condition for
\eqref{P2c} is
\begin{equation}
\sum_{j\in\mathcal{J}}\sum_{k=1}^{K}
\vartheta^{(t)}(c_{j,k};\mathbf{t}_{j,k};\mathbf{s})
\ge \Omega^{-1}(E_h).
\label{eq:eh_sca}
\end{equation}

In the ideal continuous-phase model, the RIS coefficients satisfy $|s_m|=1$, $\forall m$.
Following the penalty-based approach \cite{Kumar}, we relax $|s_m|=1$ to $|s_m|\le 1$ and add a penalty
term to encourage $|s_m|$ approaches $1$.
Given $\mathbf{s}^{(t-1)}$, we solve the following convex problem
\begin{subequations}\label{Pris_sca}
\begin{align}
\max_{\boldsymbol{\Xi}}
\quad
&
\zeta
+
C\!\left(
2\,\mathrm{Re}\!\left\{(\mathbf{s}^{(t-1)})^{H}\mathbf{s}\right\}
-\|\mathbf{s}^{(t-1)}\|^2
\right)
\label{Pris_sca_obj}
\\
\text{s.t.}\quad
&
\eqref{eq:ris_zf_affine_a}-\eqref{eq:ris_zf_affine_e},
\eqref{eq:eh_sca},
\label{Pris_sca_eh}
\\
&
|s_m|\le 1,\quad \forall m=1,\dots,M,
\label{Pris_sca_mod}
\end{align}
\end{subequations}
where the optimization variable set is defined as
$\boldsymbol{\Xi}
\triangleq
\big\{
\mathbf{s},\,
\zeta,\,
\boldsymbol{\rho},\,
\boldsymbol{\rho}_E,\,
\mathbf{f}_E,\,
\boldsymbol{\xi}
\big\}$ with
$\boldsymbol{\rho} \triangleq \{\rho_k\}_{k\in\mathcal{K}}$,
$\boldsymbol{\rho}_E \triangleq \{\rho_{E,j,k}\}_{j\in\mathcal{J},\,k\in\mathcal{K}}$,
$\mathbf{f}_E \triangleq \{f_{E,k}\}_{k\in\mathcal{K}}$ and
$\boldsymbol{\xi} \triangleq \{\xi_{j,k}\}_{j\in\mathcal{J},\,k\in\mathcal{K}}$.

\noindent The RIS phase optimization procedure is summarized in Algorithm~\ref{Alg:ris_sca}.

\begin{algorithm}[t]
\caption{SCA-Based RIS Phase Optimization (Given $\tilde{\mathbf{p}}$ and $\mathbf{q}$)}
\label{Alg:ris_sca}
\KwIn{Initial $\mathbf{s}^{(0)}$, tolerance $\varepsilon>0$, penalty weight $C>0$.}
\KwOut{Optimized RIS phase vector $\mathbf{s}^\star$.}
Set $t\gets 1$.\;
\Repeat{$|\zeta^{(t)}-\zeta^{(t-1)}|\leq \varepsilon$}{
Solve \eqref{Pris_sca} to obtain $(\mathbf{s}^{(t)},\zeta^{(t)})$.\;
$t\gets t+1$.\;
}
\Return $\mathbf{s}^{(t-1)}$.\;
\end{algorithm}


\subsection{UAV Location Optimization}
\label{subsec:sca_uav}
In this subsection, we optimize the UAV horizontal position $\mathbf{q}$
while keeping the BS precoders $\{\mathbf{p}_k\}_{k\in\mathcal{K}}$, and the RIS phase
matrix $\boldsymbol{\Theta}$ fixed.
Under fixed $\{\mathbf{p}_k\}$, the BS transmit power is constant; hence, the SEE
maximization reduces to maximizing the minimum secrecy rate. 
To obtain a numerically stable and low-complexity SCA benchmark, we further consider the relaxed objective of maximizing the minimum rate of IRs instead of the secrecy rate.

Using the distance definitions in \eqref{eq:distance1}--\eqref{eq:distance3}, the
channels can be expressed as
\begin{align}
\mathbf{g}_k(\mathbf{q})
&=
\sqrt{\rho_0}\,d_k(\mathbf{q})^{-\frac{\alpha}{2}}\,\tilde{\mathbf{g}}_k,
\qquad \forall k\in\mathcal{K},\\
\mathbf{h}_j(\mathbf{q})
&=
\sqrt{\rho_0}\,d_j(\mathbf{q})^{-\frac{\alpha}{2}}\,\tilde{\mathbf{h}}_j,
\qquad \forall j\in\mathcal{J},\\
\mathbf{G}_b(\mathbf{q})
&=
\sqrt{\rho_0}\,d_b(\mathbf{q})^{-\frac{\alpha}{2}}\,\tilde{\mathbf{G}}_b .
\end{align}
where $\alpha = \alpha_{b}=\alpha_{k}=\alpha_{j}$ for notational simplicity. Define the following fixed gains that incorporate the beamforming vectors
$\{\mathbf{p}_k\}$:
\begin{align}
\tilde{A}_{k,\ell}
&\triangleq
\rho_0^2\Big|
\tilde{\mathbf{g}}_k^H \boldsymbol{\Theta}\tilde{\mathbf{G}}_b \mathbf{p}_\ell
\Big|^2,
\qquad \forall k,\ell\in\mathcal{K},\\
\tilde{D}_{j,\ell}
&\triangleq
\rho_0 \tilde{\mathbf{h}}_j^H \boldsymbol{\Theta}\tilde{\mathbf{G}}_b \mathbf{p}_\ell,
\qquad \forall j\in\mathcal{J},\ \forall \ell\in\mathcal{K}.
\end{align}
Since both the desired signal and interference powers scale with the product
$(d_k(\mathbf q)d_b(\mathbf q))^{-\alpha}$, it is convenient to decouple the
distance-dependent terms. To this end, we introduce slack variables
$\{y_{kb}\}$ and $\{y_{jb}\}$ such that
\begin{subequations}\label{eq:y_slacks_uav}
\begin{align}
d_k(\mathbf{q})\,d_b(\mathbf{q}) &\le y_{kb}^{\frac{1}{\alpha}},
\qquad \forall k\in\mathcal{K},\label{eq:ykb_uav}\\
d_j(\mathbf{q})\,d_b(\mathbf{q}) &\le y_{jb}^{\frac{1}{\alpha}},
\qquad \forall j\in\mathcal{J}.
\label{eq:yjb_uav}
\end{align}
\end{subequations}
Although the function $y^{1/\alpha}$ is concave for $\alpha \ge 2$,
the constraints in \eqref{eq:y_slacks_uav} remain nonconvex due to the bilinear
terms $d_k(\mathbf q)d_b(\mathbf q)$ and $d_j(\mathbf q)d_b(\mathbf q)$.
To construct a convex surrogate, we employ the arithmetic--geometric mean (AGM) 
inequality. 
Specifically, at iteration $t$, for any nonnegative variables 
$a$ and $b$, the product $ab$ is upper bounded as
\begin{equation}
ab \le \widetilde{\Theta}^{(t)}( a,b)
\triangleq
\frac{1}{2}\!\left(
\frac{b^{(t-1)}}{a^{(t-1)}}a^2
+
\frac{a^{(t-1)}}{b^{(t-1)}}b^2
\right),
\end{equation}
where $(a^{(t-1)},b^{(t-1)})$ denotes the point from the previous iteration. The surrogate $\widetilde{\Theta}^{(t)}( a,b)$ is convex and tight at $(a^{(t-1)},b^{(t-1)})$.
Consequently, the slack-variable constraints admit the following 
convex SCA approximation
\begin{align}
\widetilde{\Theta}^{(t)}( d_k(\mathbf q),d_b(\mathbf q))
\le y_{kb}^{\frac{1}{\alpha}},\ \forall k,
\label{eq:y_safe_bound_uav}
\\
\widetilde{\Theta}^{(t)}( d_j(\mathbf q),d_b(\mathbf q))
\le y_{jb}^{\frac{1}{\alpha}},\ \forall j.
\label{eq:y_safe_bound_uav2}
\end{align}
By multiplying both the numerator and denominator of the SINR expression by
$(d_k(\mathbf q)d_b(\mathbf q))^{\alpha}$ and using the slack variable
definition in \eqref{eq:y_slacks_uav}, the legitimate-user SINR can be
conservatively represented as
\begin{equation}\label{eq:legit_sinr_uav}
\gamma_k(\mathbf{q})
=
\frac{\tilde{A}_{k,k}}
{\sum_{\ell\neq k}\tilde{A}_{k,\ell}+\sigma^2\,y_{kb}},
\qquad \forall k\in\mathcal{K}.
\end{equation}


Introduce auxiliary SINR variables $\{\tilde \rho_k\}$ and the rate variable
$\zeta$. Using \eqref{eq:legit_sinr_uav}, the rate constraints can be equivalently
written as
\begin{subequations}\label{eq:uav_rate_aux}
\begin{align}
&\zeta \le \log_2(1+\tilde \rho_k),\qquad \forall k\in\mathcal K,
\label{eq:uav_rate_aux_a}
\\&
\tilde \rho_k\!\left(\sum_{\ell\neq k}\tilde A_{k,\ell}+\sigma^2 y_{kb}\right) \le \tilde A_{k,k},
\qquad \forall k\in\mathcal K.
\label{eq:uav_rate_aux_b}
\end{align}
\end{subequations}
The bilinear term $\tilde \rho_k y_{kb}$ in \eqref{eq:uav_rate_aux_b} is nonconvex.
At iteration $t$, it is upper bounded by the convex AGM inequality introduced in the previous subsection as
\begin{equation}\label{eq:uav_agm_majorizer}
\tilde \rho_k y_{kb}\le
\widetilde{\Theta}^{(t)}(\tilde \rho_k,y_{kb})
\end{equation}
which is convex in $(\tilde \rho_k,y_{kb})$ and tight at 
$(\tilde \rho_k^{(t-1)},y_{kb}^{(t-1)})$.
 Hence, a convex surrogate of
\eqref{eq:uav_rate_aux_b} is
\begin{equation}\label{eq:uav_rate_convex}
\tilde \rho_k\sum_{\ell\neq k}\tilde A_{k,\ell}
+\sigma^2\,\widetilde{\Theta}^{(t)}(\tilde \rho_k,y_{kb})
\le \tilde A_{k,k},\qquad \forall k\in\mathcal K.
\end{equation}

The harvested RF power at UEHR~$j$ is
\begin{equation}\label{eq:eh_uav}
P_j^{\mathrm{EH}}(\mathbf q)
=
\sum_{\ell=1}^{K}\big|c_{j,\ell}+x_{jb}\tilde D_{j,\ell}\big|^2,
\qquad \forall j\in\mathcal J,
\end{equation}
where
\(
x_{jb}\triangleq\big(d_j(\mathbf q)d_b(\mathbf q)\big)^{-\alpha/2}
\)
is the cascaded path-loss factor. The nonconvexity of \eqref{eq:eh_uav} stems from
the distance-dependent term $x_{jb}$.
To obtain a numerically stable SCA formulation, we eliminate $x_{jb}$ 
by expressing the bound in terms of $x_{jb}^2$ and using the relation 
$x_{jb}^2 \ge 1/y_{jb}$, as established in Proposition~\ref{prop:eh_inv_y_lb}.

\begin{Proposition}\label{prop:eh_inv_y_lb}
Consider \eqref{eq:eh_uav} 
and let $y_{jb}$ be defined as in \eqref{eq:yjb_uav}.
Define the constants
$U_j\triangleq \sum_{\ell=1}^{K}|c_{j,\ell}|^2$,
$S_j\triangleq \sum_{\ell=1}^{K}|\tilde D_{j,\ell}|^2$, and
$T_j\triangleq \sum_{\ell=1}^{K}\Re\!\left\{c_{j,\ell}^\ast \tilde D_{j,\ell}\right\}$.
Then, for any $\varepsilon_j>1/S_j$ (when $S_j>0$), the harvested power admits the
global lower bound
\begin{equation}\label{eq:eh_lb_inv_y_main}
P_j^{\mathrm{EH}}(\mathbf q)
\ \ge\
A_j(\varepsilon_j)+B_j(\varepsilon_j)\,\frac{1}{y_{jb}},
\qquad \forall j\in\mathcal J,
\end{equation}
where $A_j(\varepsilon_j)\triangleq U_j-\varepsilon_j T_j^2$, and $B_j(\varepsilon_j)\triangleq S_j-\frac{1}{\varepsilon_j}>0$.
\end{Proposition}
\begin{proof}
    See Appendix \ref{app:eh_inv_y_lb}.
\end{proof}
Using \eqref{eq:eh_lb_inv_y_main}, the EH requirement can be enforced in SCA by
linearizing $1/y_{jb}$ around the previous iterate.
Since $h(y)=1/y$ is convex for $y>0$, its first-order Taylor approximation at
$y_{jb}^{(t-1)}$ is a global affine underestimator
\begin{equation}\label{eq:inv_taylor_app}
\frac{1}{y_{jb}}
\ge
\frac{1}{y_{jb}^{(t-1)}}-\frac{1}{\big(y_{jb}^{(t-1)}\big)^2}\big(y_{jb}-y_{jb}^{(t-1)}\big)
\triangleq \widetilde\psi_j^{(t)}(y_{jb}).
\end{equation}
Therefore, a sufficient convex inner approximation of the EH constraint \eqref{P2c} is
\begin{equation}\label{eq:eh_sca_final_app}
\sum_{j \in \mathcal{J}} \left ( A_j(\varepsilon_j)+B_j(\varepsilon_j)\,\widetilde\psi_j^{(t)}(y_{jb}) \right )
\ \ge\ \Omega^{-1}(E_h).
\end{equation}
\begin{Remark}
A simple and stable choice of $\varepsilon_j$ is $\varepsilon_j=\frac{2}{S_j}$ (for $S_j>0$), which
yields $B_j(\varepsilon_j)=S_j/2$ and $A_j(\varepsilon_j)=U_j-2T_j^2/S_j$.
\end{Remark}
At SCA iteration $t$, the UAV location is updated by solving
\begin{subequations}\label{P_uav_legit_sca}
\begin{align}
\max_{\boldsymbol{\tilde{\Xi}}}\quad & \zeta
\\
\text{s.t.}\quad
& \eqref{eq:uav_rate_aux_a},\ \eqref{eq:uav_rate_convex},\ \eqref{eq:y_safe_bound_uav}, \eqref{eq:y_safe_bound_uav2},\eqref{eq:eh_sca_final_app}, \eqref{P1h}.
\end{align}
\end{subequations}
Here the optimization variable set is
$\boldsymbol{\tilde{\Xi}}
\triangleq \{\mathbf q,\,\zeta,\,\tilde {\boldsymbol \rho},\,\mathbf y_b,\,\mathbf y_u \}$, where
$\tilde{\boldsymbol \rho} =\{\tilde \rho_k\}_{k\in\mathcal{K}}$,
$\mathbf{y}_b=\{y_{kb}\}_{k=1}^{K}$, and
$\mathbf{y}_u=\{y_{jb}\}_{j=1}^{J}$.
The proposed method for UAV location optimization is summarized in Algorithm \ref{Alg:uav_sca_legit}.  Problem \eqref{P_uav_legit_sca} is convex and can be efficiently solved using standard convex optimization solvers (e.g., CVX).
The linearization points
$\{\tilde \rho_k^{(t)},y_{kb}^{(t)},y_{jb}^{(t)}\}$ are updated iteratively after each  solve until convergence.

\begin{algorithm}[t]
\caption{SCA-Based UAV Location Optimization}
\label{Alg:uav_sca_legit}
\KwIn{Initial $\mathbf q^{(0)}\in\mathcal L$, tolerance $\varepsilon>0$.}
\KwOut{Optimized UAV location $\mathbf q^\star$.}
Set $t\gets 1$ and initialize $(\tilde{\boldsymbol\rho}^{(0)},\mathbf y_b^{(0)},\mathbf y_u^{(0)})$ from $\mathbf q^{(0)}$.\;
\Repeat{$|\zeta^{(t)}-\zeta^{(t-1)}|\le \varepsilon$}{
Construct $\widetilde{\Theta}^{(t)}(\cdot)$
 and $\widetilde{\psi}_j^{(t)}(\cdot)$ using
\eqref{eq:uav_agm_majorizer} and \eqref{eq:inv_taylor_app}.\;
Solve \eqref{P_uav_legit_sca} to obtain $(\mathbf q^{(t)},\zeta^{(t)})$.\;
Update  the linearization points and set $t\gets t+1$.\;
}
\Return $\mathbf q^{(t-1)}$.\;
\end{algorithm}

By alternately optimizing the power allocation, RIS reflection phases, and UAV
location, the benchmark problem in \eqref{P2} can be efficiently solved via a
BCD framework, where each block is handled using the
SCA-based procedures developed in
Sections~\ref{subsec:sca_power}, \ref{subsec:sca_ris}, and \ref{subsec:sca_uav}.
The complete procedure is summarized in Algorithm~\ref{alg:overall_zf}.

\begin{algorithm}[t]
\caption{BCD--SCA Benchmark for WCSEE}
\label{alg:overall_zf}
\DontPrintSemicolon
\LinesNumbered
\KwIn{Init. $\tilde{\mathbf p}^{(0)}\!\succeq\!0$, $\mathbf s^{(0)}$, $\mathbf q^{(0)}\!\in\!\mathcal L$; 
$\epsilon_{\rm out}$, $I_{\max}$}
\KwOut{$\tilde{\mathbf p}^\star$, $\mathbf s^\star$, $\mathbf q^\star$}

$i\!\gets\!0$, $\eta^{(0)}\!\gets\!0$.\;
\Repeat{$|\eta^{(i)}-\eta^{(i-1)}|\le\epsilon_{\rm out}$ \textbf{or} $i=I_{\max}$}{
Update $\mathbf H_c$ and compute ZF directions $\{\hat{\mathbf p}_k^{(i)}\}$.\;

$\tilde{\mathbf p}^{(i+1)}\!\leftarrow\!$ Alg.~\ref{Alg:pow_cc_sca}.\;
$\mathbf s^{(i+1)}\!\leftarrow\!$ Alg.~\ref{Alg:ris_sca}.\;
$\mathbf q^{(i+1)}\!\leftarrow\!$ Alg.~\ref{Alg:uav_sca_legit}.\;

Compute $R^{\rm sec,(i+1)}$ and update
$\eta^{(i+1)}\!=\!\frac{R^{\rm sec,(i+1)}}{\varrho\sum_k \tilde p_k^{(i+1)}+P_0}$.\;

$i\!\gets\!i+1$.\;
}
\Return $\tilde{\mathbf p}^{(i)}, \mathbf s^{(i)}, \mathbf q^{(i)}$.\;
\end{algorithm}

\section{SAC-Based Joint ZF Power Allocation, RIS Phase, and Location Optimization}\label{SAC-section}

Due to the nonconvexity of the WCSEE maximization problem in Section~\ref{Problem-formulation}, conventional AO and SCA-based methods may suffer from poor scalability and convergence to local optima.
To address these challenges, we adopt a SAC framework that learns control policies through interaction with the wireless environment. SAC effectively handles high-dimensional continuous action spaces, ensures stable learning via entropy regularization, and accommodates discrete RIS phase constraints through continuous relaxation and projection. 
In the following, the problem is first reformulated as a Markov decision process (MDP), after which the SAC-based solution is developed.

\subsection{MDP Model}
To enable a learning-based solution, the original optimization problem is reformulated as a sequential decision-making process. The joint design variables, including UAV location, RIS phase shifts, and transmit power, are optimized through interaction with the environment. At each time step, the agent observes the current network configuration (e.g., channel conditions and UAV position) and updates the control variables to improve system performance. This naturally leads to an MDP formulation, where the system evolution depends on the current state and action.

The problem is thus modeled as an MDP characterized by the tuple
$\langle\mathcal{S},\mathcal{A},\mathcal{P},\mathcal{R}\rangle$, where the state represents the 
instantaneous channel conditions and spatial deployment of the UAV-mounted RIS, the action corresponds to the joint control of UAV location, RIS phase shifts, and transmit power, and the reward reflects the achieved WCSEE performance. At each step $t$, the agent observes $s_t\in\mathcal{S}$, selects an action $a_t\in\mathcal{A}$ according to a policy $\pi(a_t|s_t)$, and receives a reward $r_t=\mathcal{R}(s_t,a_t)$. The components are defined as follows.
\begin{itemize}
\item \textbf{Action:} The action includes ZF power allocation, RIS phase shifts, and UAV horizontal position:
\begin{equation}
a_t
=
\big[
\{\tilde{p}_k^{(t)}\}_{k=1}^{K},
\{\theta_m^{(t)}\}_{m=1}^{M},
\mathbf{q}^{(t)}
\big].
\end{equation}
The SAC agent outputs a continuous vector in $[-1,1]^{K+M+2}$ via a \emph{tanh} layer, which is mapped to feasible variables. Power is first scaled to $[0,1]$ and normalized to satisfy \eqref{P1d}:
\begin{equation}
\tilde p_k^{(t)}
=
P_{\max}\,
\frac{\hat p_k^{(t)}}{\sum_{\ell=1}^{K}\hat p_\ell^{(t)}},
\quad \forall k\in\mathcal K,
\end{equation}
with $\sum_{\ell=1}^{K}\hat p_\ell^{(t)}>0$. The RIS phases are mapped to $[0,2\pi)$ and quantized to $\mathcal F$, while $\mathbf q^{(t)}$ is projected onto $\mathcal L$.

\item \textbf{State:} The state captures channel and geometric information:
\begin{align}
s_t =
\big[
&
\mathfrak{Re}\{\mathbf{H}_{\mathrm{c}}\},\mathfrak{Im}\{\mathbf{H}_{\mathrm{c}}\},
\mathfrak{Re}\{\hat{\mathbf{u}}_j\}_{j=1}^J,\mathfrak{Im}\{\hat{\mathbf{u}}_j\}_{j=1}^J, \nonumber\\
&
\mathbf{q}^{(t-1)}
\big].
\end{align}
Here, $\mathbf{H}_{\mathrm{c}}$ and $\{\hat{\mathbf{u}}_j\}_{j=1}^J$ denote the current effective cascaded channels corresponding to the RIS phase configuration determined by the previous action $a_{t-1}$, while $\mathbf{q}^{(t-1)}$ denotes the UAV location resulting from the previous action.
The user and UEHR locations remain fixed within each episode and vary across episodes.

\item \textbf{Reward:} Given $a_t$, the ZF precoder is constructed and the worst-case eavesdropping SINR $\gamma_{E,k}^{(t)}$ is computed via \eqref{eq:wc_eve_sinr}. The total power is
$P_{\mathrm{tot}}^{(t)}=\varrho \sum_{k=1}^{K}\tilde{p}_k^{(t)} + P_0$. The reward is
\begin{equation}
r_t
=
\begin{cases}
\dfrac{R^{\mathrm{sec},(t)}}{P_{\mathrm{tot}}^{(t)}}, & \text{if \eqref{P1c_closed} holds},\\
0, & \text{otherwise},
\end{cases}
\end{equation}
where $R^{\mathrm{sec},(t)}$ is defined in \eqref{eq:min_sec_rate}.  This feasibility-based reward design ensures that only constraint-satisfying actions contribute to the learning process, preventing the agent from pursuing a trivial policy that achieves high secrecy by reducing or cutting off energy transfer to the UEHRs. Consequently, it guides the SAC agent toward policies that satisfy the mandatory EH operational constraint while maximizing the WCSEE, consistent with constraint-aware reward shaping widely adopted in the literature (see, e.g., \cite{Niyato,ZhangRSMA}).
\end{itemize}


\subsection{Fundamentals of the SAC Algorithm}

Unlike conventional DRL algorithms that only maximize the expected cumulative reward, the SAC algorithm augments the objective with an entropy term to encourage policy exploration and improve training stability. The optimal policy $\pi^\star$ is obtained by solving the maximum-entropy reinforcement learning (RL) problem \cite{Haarnoja}
\begin{equation}
\pi^\star
=
\arg\max_{\pi}
\mathbb{E}_{(s_t,a_t)\sim\rho_\pi}
\Bigg[
    \sum_{t=0}^{\infty}
    \big(
        r_t + \beta \mathcal{H}(\pi(\cdot|s_t))
    \big)
\Bigg],
\label{eq:sac_objective}
\end{equation}
where $\rho_\pi$ denotes the state--action visitation distribution under policy $\pi$, $\mathcal{H}(\pi(\cdot|s_t)) \triangleq -\mathbb{E}_{a_t\sim\pi(\cdot|s_t)}[\log\pi(a_t|s_t)]$ denotes the policy entropy, and $\beta>0$ is the temperature parameter that balances exploitation and exploration.

In SAC, two soft Q-functions $Q_{\phi_1}(s,a)$ and $Q_{\phi_2}(s,a)$, parameterized by neural networks with parameters $\phi_1$ and $\phi_2$, are employed to approximate the soft state--action value:
\begin{equation}
Q^\pi(s_t,a_t) =
\mathbb{E}
\big[
    r_t + \gamma\, V^\pi(s_{t+1})
\big],
\end{equation}
where  $\gamma \in [0,1]$ is the discount factor weighting future and
instant rewards and
the soft value function is given by
\begin{equation}
V^\pi(s_{t+1})
=
\mathbb{E}_{a_{t+1}\sim\pi}
\!\left[
    Q^\pi(s_{t+1},a_{t+1})
    -
    \beta\log\pi(a_{t+1}|s_{t+1})
\right].
\end{equation}
In practice, the soft value function is not explicitly parameterized and is
instead evaluated implicitly using the current policy and target Q-networks, in
accordance with the modern SAC formulation.
Given a replay buffer $\mathcal{D}$ containing transition tuples $(s_t,a_t,r_t,s_{t+1})$, the critic networks are trained by minimizing the following mean squared error (MSE)  soft Bellman loss
\begin{equation} \label{Bellman error}
J_Q(\phi_i) =
\mathbb{E}_{(s_t,a_t,r_t,s_{t+1})\sim\mathcal{D}}
\Bigg[
\frac{1}{2}
\Big(
    Q_{\phi_i}(s_t,a_t)
    -
    y_t
\Big)^2
\Bigg],
\end{equation}
where $i\in\{1,2\}$ and
\begin{align}
y_t
=\, & r_t
+ \gamma\,\mathbb{E}_{a_{t+1}\sim\pi(\cdot|s_{t+1})}
\nonumber\\ &
\left[
    \min_{j=1,2} Q_{\bar{\phi}_j}(s_{t+1},a_{t+1})
    -
    \beta \log \pi(a_{t+1}|s_{t+1})
\right],
\end{align}
with $\bar{\phi}_j$ denoting the parameters of the corresponding target Q-networks.

The actor (policy) network $\pi_\psi(a|s)$, parameterized by $\psi$, is updated by minimizing
\begin{align}\label{policy update}
J_\pi(\psi)
=\, & \mathbb{E}_{s_t\sim\mathcal{D}}
\Big[
    \mathbb{E}_{a_t\sim\pi_\psi}
   \nonumber \\& \big(
        \beta\log\pi_\psi(a_t|s_t)
        - \min_{i=1,2}Q_{\phi_i}(s_t,a_t)
    \big)
\Big].
\end{align}
The temperature $\beta$ is tuned automatically by minimizing \cite{Haarnoja2}
\begin{equation} \label{temperature}
J(\beta) =
\mathbb{E}_{(s_t,a_t)\sim\mathcal{D}}
\Big[
    -\beta
    \big(
        \log\pi_\psi(a_t|s_t) + \bar{\mathcal{H}}
    \big)
\Big],
\end{equation}
where $\bar{\mathcal{H}}$ is a target entropy level, typically set as a negative constant proportional to the action dimension.

Beyond this practical perspective, SAC admits a useful interpretation based on probabilistic inference: the policy improvement step corresponds to projecting the policy distribution toward an energy--based target distribution. Formally, the optimal policy solves
\begin{equation}
J_{\pi}(\psi)
=
\mathbb{E}_{s_t\sim\mathcal{D}}
\!\left[
D_{\mathrm{KL}}
\!\left(
\pi_{\psi}(\cdot|s_t)
\,\middle\|
\,\frac{\exp\!\left(Q_{\phi}(s_t,\cdot)/\beta\right)}{
Z_{\phi}(s_t)}
\right)
\right],
\label{eq:kl_policy_update}
\end{equation}
where $D_{\mathrm{KL}}(\cdot\|\cdot)$ denotes the Kullback--Leibler divergence and $Z_{\theta}(s_t)$ is a normalizing partition function. This viewpoint reveals that SAC implicitly projects the policy toward a Boltzmann distribution induced by the critic values, which explains its stable convergence and robust exploration behaviour.


\subsection{SAC-Based Joint Optimization Framework}
Unlike deterministic actor--critic baselines, SAC updates the policy via
an entropy-regularized objective without a target actor network.
Fig.~\ref{SAC-BD} depicts the interaction between the SAC agent and the wireless
environment.  Moreover, Fig.~\ref{fig:sac_architecture} illustrates the architecture of the proposed SAC networks, where the actor employs fully connected (FC) layers with rectified linear unit (ReLU) activations to output the mean and log-standard deviation of a Gaussian policy, followed by Gaussian sampling and $\tanh$ squashing, while each critic takes the concatenated state--action pair as input and outputs a scalar Q-value. The training protocol follows an off-policy paradigm with experience replay and mini-batch updates, where the actor, twin critics, and entropy temperature are iteratively optimized using sampled transitions.
At each step, the agent observes the state $s_t$ and samples
a continuous action $a_t\sim\pi_\psi(\cdot|s_t)$ that jointly controls the
power allocation, RIS phase-shift vector, and UAV horizontal position.
The critic employs twin Q-networks $Q_{\phi_1}$ and $Q_{\phi_2}$, and
the minimum Q-value is used to mitigate overestimation bias. Target critic
networks $\{Q_{\bar{\phi}_1},Q_{\bar{\phi}_2}\}$ are updated via Polyak
averaging to stabilize Bellman backups, while the entropy temperature
$\beta$ is automatically tuned to match a target entropy level.
The environment maps actions to feasible variables via power normalization,
RIS phase quantization, and location projection, and returns the reward
based on the hard EH constraint and the resulting WCSEE. Transitions are stored in a replay buffer and mini-batches are sampled
to update the actor, critics, and $\beta$ as summarized in
Algorithm~\ref{alg:SAC_RIS}. After training, $\pi_\psi$ enables real-time
joint control under different channel realizations and network geometries.
\begin{figure} 
    \centering
    \includegraphics[width=\linewidth]{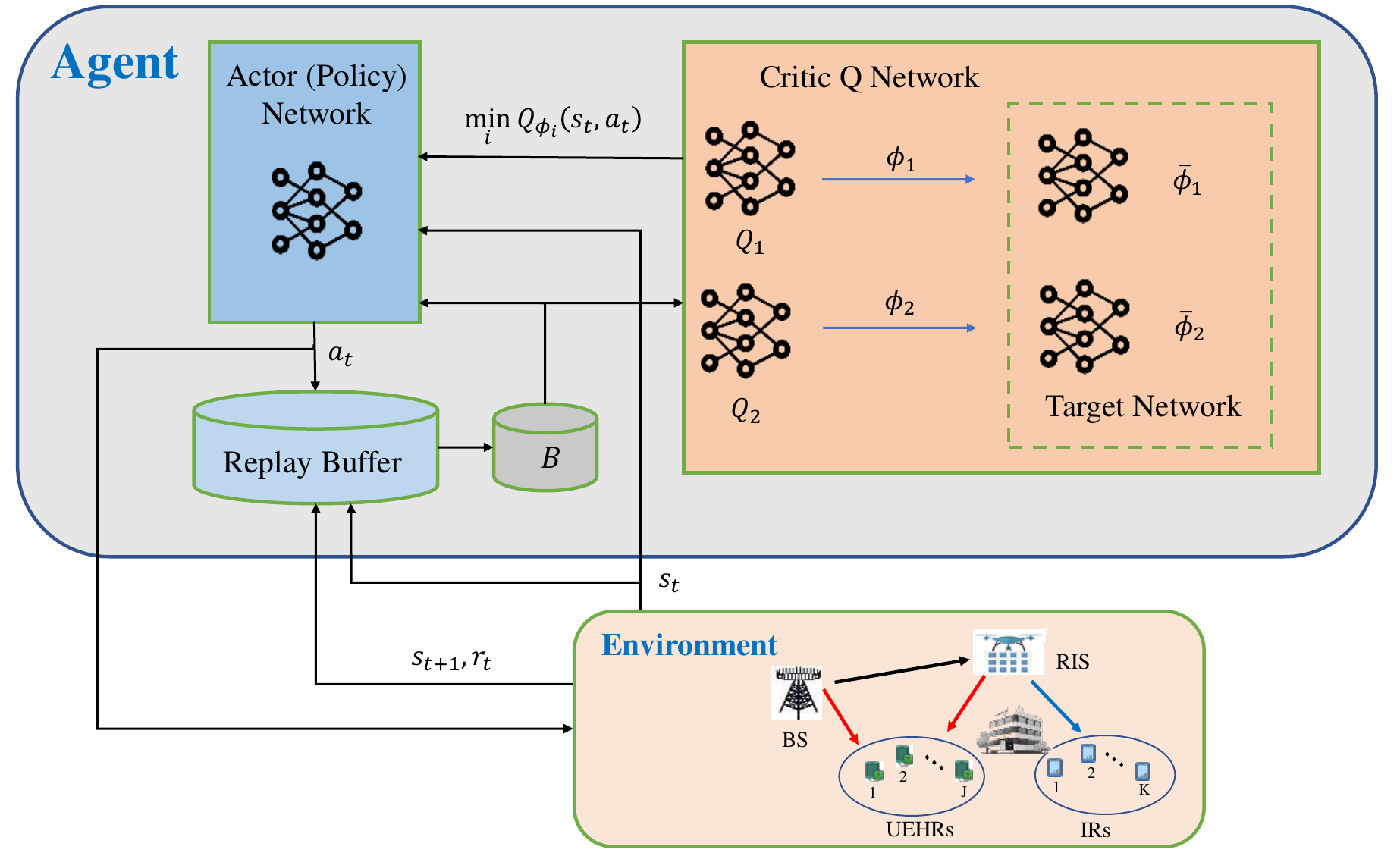}
    \caption{SAC framework for joint power allocation,
RIS phase optimization, and UAV positioning with twin critics and target
Q-networks.}
    \label{SAC-BD}		
\end{figure}

\begin{figure}[t]
\centering
\resizebox{\linewidth}{!}{%
\begin{tikzpicture}[
    font=\small,
    >=Latex,
    line width=0.45pt,
    node distance=0.55cm and 0.55cm,
    block/.style={draw, rectangle, minimum height=0.78cm, minimum width=1.38cm, align=center, fill=blue!8},
    block2/.style={draw, rectangle, minimum height=0.78cm, minimum width=1.30cm, align=center, fill=orange!10},
    block3/.style={draw, rectangle, minimum height=0.78cm, minimum width=1.38cm, align=center, fill=green!10},
    sumblock/.style={draw, rectangle, minimum height=0.78cm, minimum width=1.70cm, align=center, fill=yellow!15},
    group/.style={draw, dashed, rounded corners, inner sep=5pt},
    arr/.style={->, >=Latex, line width=0.45pt}
]

\node at (6,3.2) {\textbf{(a) Actor network (Gaussian policy)}};

\node[block]  (s)      at (0.0,1.25) {State\\$s_t$};
\node[block]  (a1)     at (1.75,1.25) {FC\\256};
\node[block2] (a2)     at (3.35,1.25) {ReLU};
\node[block]  (a3)     at (4.95,1.25) {FC\\256};
\node[block2] (a4)     at (6.55,1.25) {ReLU};

\node[sumblock] (sample) at (8.65,1.2) {Gaussian\\sampling};
\node[block3]   (mean)   at (8.65,2.28) {Mean\\head};
\node[block3]   (logstd) at (8.65,0.1) {Log-std\\head};

\node[block2] (tanh) at (10.45,1.25) {$\tanh$};
\node[block3] (act)  at (12.05,1.25) {Action\\$a_t$};

\draw[arr] (s) -- (a1);
\draw[arr] (a1) -- (a2);
\draw[arr] (a2) -- (a3);
\draw[arr] (a3) -- (a4);

\draw[arr] (a4.east) -- ++(0.5,0) |- (mean.west);
\draw[arr] (a4.east) -- ++(0.5,0) |- (logstd.west);

\draw[arr] (mean.south) -- (sample.north);
\draw[arr] (logstd.north) -- (sample.south);

\draw[arr] (sample) -- (tanh);
\draw[arr] (tanh) -- (act);

\node[group, fit=(s)(a1)(a2)(a3)(a4)(mean)(logstd)(sample)(tanh)(act)] {};

\node at (6,-0.75) {\textbf{(b) Twin critic networks}};

\node[block]    (cs)     at (0.0,-2.35) {State\\$s_t$};
\node[block3]   (ca)     at (1.65,-2.35) {Action\\$a_t$};
\node[sumblock] (concat) at (3.5,-2.35) {Concatenate\\$[s_t,a_t]$};

\node[block]  (q11)   at (5.60,-1.60) {FC\\512};
\node[block2] (q12)   at (7.20,-1.60) {ReLU};
\node[block]  (q13)   at (8.80,-1.60) {FC\\512};
\node[block2] (q14)   at (10.40,-1.60) {ReLU};
\node[block3] (q1out) at (12,-1.60) {$Q_1$};

\node[block]  (q21)   at (5.60,-3.10) {FC\\512};
\node[block2] (q22)   at (7.20,-3.10) {ReLU};
\node[block]  (q23)   at (8.80,-3.10) {FC\\512};
\node[block2] (q24)   at (10.40,-3.10) {ReLU};
\node[block3] (q2out) at (12,-3.10) {$Q_2$};

\draw[arr] (cs) -- (ca);
\draw[arr] (ca) -- (concat);

\coordinate (split) at ($(concat.east)+(0.22,0)$);
\draw[arr] (concat.east) -- (split);
\draw[arr] (split) |- (q11.west);
\draw[arr] (split) |- (q21.west);

\draw[arr] (q11) -- (q12);
\draw[arr] (q12) -- (q13);
\draw[arr] (q13) -- (q14);
\draw[arr] (q14) -- (q1out);

\draw[arr] (q21) -- (q22);
\draw[arr] (q22) -- (q23);
\draw[arr] (q23) -- (q24);
\draw[arr] (q24) -- (q2out);

\node[group, fit=(cs)(ca)(concat)(q11)(q12)(q13)(q14)(q1out)(q21)(q22)(q23)(q24)(q2out)] {};

\end{tikzpicture}%
}
\caption{Structure of the proposed SAC networks: (a) Gaussian actor network and (b) twin Q-critic networks.}
\label{fig:sac_architecture}
\end{figure}

\begin{algorithm}[t]
\caption{SAC-Based Joint Optimization}
\label{alg:SAC_RIS}
\DontPrintSemicolon
\SetKwInOut{Input}{Input}
\SetKwInOut{Output}{Output}
\Input{$\gamma,\lambda_Q,\lambda_\pi,\lambda_\beta,\bar{\mathcal H},\tau,B$}
\Output{Policy $\pi_\psi(a|s)$}
Initialize replay buffer $\mathcal D$, critics $Q_{\phi_1},Q_{\phi_2}$, targets $Q_{\bar\phi_1},Q_{\bar\phi_2}$ with $\bar\phi_i\!\leftarrow\!\phi_i$, actor $\pi_\psi$, and temperature $\beta$.\;
\For{each episode}{
Reset environment and observe $s_0$.\;
\For{each step $t$}{
Sample $a_t \sim \pi_\psi(\cdot|s_t)$, map it to $\{\tilde{p}_k^{(t)}\},\{\theta_m^{(t)}\},\mathbf{q}^{(t)}$, execute it, and obtain $(r_t,s_{t+1})$.\;
Store $(s_t,a_t,r_t,s_{t+1})$ in $\mathcal D$.\;
Sample a mini-batch $\{(s,a,r,s')\}$ from $\mathcal D$.\;
Update critics: $\phi_i\!\leftarrow\!\phi_i-\lambda_Q\nabla_{\phi_i}J_Q(\phi_i),~i=1,2$.\;
Update actor: $\psi\!\leftarrow\!\psi-\lambda_\pi\nabla_\psi J_\pi(\psi)$.\;
Update temperature: $\beta\!\leftarrow\!\beta-\lambda_\beta\nabla_\beta J(\beta)$.\;
Soft-update targets: $\bar\phi_i\!\leftarrow\!\tau\phi_i+(1-\tau)\bar\phi_i,~i=1,2$.\;
}
}
\end{algorithm}


\subsection{Computational Complexity}

The computational complexity of the proposed SAC-based framework consists of two components: i) environment interaction and reward computation, and ii) SAC network updates.

At each decision step, mapping the action to $\tilde{\mathbf p}^{(t)}$, $\boldsymbol{\theta}^{(t)}$, and $\mathbf q^{(t)}$ requires $O(K+M)$ operations. Constructing the cascaded channels for $K$ IRs and $J$ UEHRs incurs $O\big((K+J)N_tM\big)$ complexity. The ZF precoder computation requires $O(N_tK^2+K^3)$, while secrecy-rate and eavesdropping evaluation costs $O(JKN_t+JK^2)$. Hence, the per-step environment complexity is
\begin{equation}
C_{\rm env}
=
O\big((K+J)N_tM + N_tK^2 + K^3 + JKN_t + JK^2\big).
\end{equation}

The SAC algorithm employs one actor and two critic networks. Let $P_\pi$ and $P_Q$ denote the number of parameters of the actor and each critic, respectively. For a minibatch of size $B$, each update requires
\begin{equation}
C_{\rm SAC}
=
O\big(B(2P_Q + P_\pi)\big).
\end{equation}
For an episode of $T$ steps, the overall training complexity is
\begin{equation}
O\!\left(T\big[C_{\rm env} + C_{\rm SAC}\big]\right).
\end{equation}
Despite the higher offline training cost, the online decision-making complexity is low, requiring only a single actor forward pass, i.e., $O(P_\pi)$, per step.

For the SCA--BCD benchmark, the complexity is governed by the outer BCD iterations and the internal iterations for solving the convex subproblems of power allocation, RIS phase optimization, and UAV location. Using a standard interior-point worst-case approximation, their complexities scale as $O\!\left(I_{\rm p}(K+2)^3\right)$, $O\!\left(I_{\rm r}(2M+2JK+2K+1)^3\right)$, and $O\!\left(I_{\rm q}(2K+J+3)^3\right)$, respectively, where $I_{\rm p}$, $I_{\rm r}$, and $I_{\rm q}$ denote the corresponding iteration numbers. Hence, the overall complexity of Algorithm~\ref{alg:overall_zf} is
\begin{align}
C_{\rm SCA}
=&
O\!\Big(
I_{\rm BCD}\big[
I_{\rm p}(K+2)^3
+ I_{\rm r}(2M+2JK+2K+1)^3 \nonumber \\&
+ I_{\rm q}(2K+J+3)^3
\big]
\Big),
\end{align}
where $I_{\rm BCD}$ denotes the number of outer BCD iterations. The cost of channel construction and ZF beamforming is omitted as it is dominated by the convex subproblem solves. Compared with the proposed SAC framework, the SCA--BCD method incurs significantly higher online complexity due to per-realization iterative optimization.


\section{Numerical Results}\label{Simulat}
This section evaluates the performance of the proposed SAC-based algorithm and compares it with SCA- and 
DRL-based benchmarks, including DDPG and TD3, under various system settings. The results demonstrate that the proposed SAC approach consistently achieves superior performance in terms of convergence behavior and average reward, highlighting its effectiveness in addressing the considered optimization problem.

\subsection{Simulation and Benchmark Settings}
Unless otherwise specified, simulations are conducted in a
$1500\times3000~\mathrm{m}^2$ area with the BS located at $(0,0)$ and a
UAV-mounted RIS initially positioned at $(1000,0,100)$~m. The UAV moves horizontally
within a $50\times50~\mathrm{m}^2$ region at a fixed altitude of $H=100$~m.
The BS employs $N_t=6$ antennas to serve $K=4$ legitimate users in the
presence of $J=3$ UEHRs, randomly distributed within circles of radius
$500$~m. The RIS consists of $M=10$ elements with discrete phase resolution
$L=8$.
Unless stated otherwise, $P_{\mathrm{max}}=10$~dBm, $\sigma^2=10^{-3}$,
$P_0=1$~W, and 
$\alpha_b = \alpha_k = \alpha_j = \alpha_{bj} =2.5$. Both perfect and ICSI are considered;
for ICSI, the uncertainty radius is $\nu\in[0.01,0.1]$, and the minimum
harvested energy requirement is $E_h\in[0.005,0.02]$.
For SCA, results are averaged over $100$ channel realizations with tolerance
$\epsilon=0.01$. For DRL-based methods, each run consists of
$T=2\times10^4$ interaction steps, and results are averaged over multiple
random seeds. The main hyperparameters are summarized in
Table~\ref{tab:drl_hyper}.

\begin{table}[t]
\centering
\caption{Hyperparameters for DRL-Based Algorithms}
\label{tab:drl_hyper}
\begin{tabular}{lll}
\hline
\textbf{Parameter} & \textbf{Description} & \textbf{Value} \\
\hline
Loss & Critic loss function & MSE \\
$\gamma$ & Discount factor & $0.99$ \\
$\tau$ & Target-network soft update factor & $5\times10^{-3}$ \\
$B$ & Mini-batch size & $256$ \\
$D$ & Replay buffer size & $10^5$ \\
$T$ & Training steps per episode & $2\times10^4$ \\
Activation & Hidden-layer activation & ReLU \\
Hidden layers & Number of hidden layers & $2$ \\
Hidden units & Units per hidden layer & $256$ \\
$\lambda_{\pi}$ & Actor learning rate & $10^{-4}$ \\
$\lambda_{Q}$ & Critic learning rate & $10^{-3}$ \\
$\lambda_{\beta}$ & Temperature learning rate & $10^{-3}$ (SAC) \\
\hline
\end{tabular}
\end{table}

\subsection{Performance in the Ideal Case}
We first evaluate the performance of the SCA, DDPG,  TD3
 and the
proposed SAC algorithms under an ideal setting with PCSI and continuous RIS phase
shifts, which serves as a performance benchmark under ideal conditions.

\begin{figure}

    \centering
    \includegraphics[width=\linewidth]{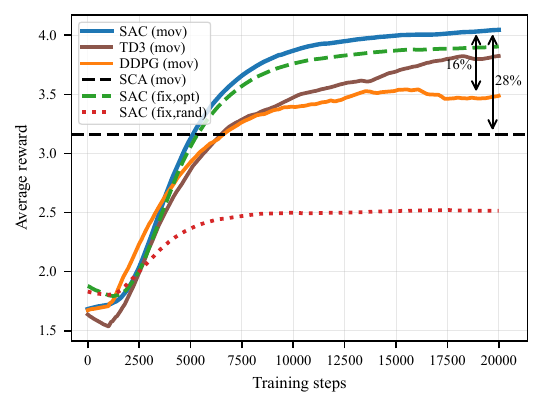}
    \caption{Convergence behavior under the ideal setting at
    $P_{\mathrm{max}}=10$~dBm.}
    \label{fig1-sim}
\end{figure}

As shown in Fig.~\ref{fig1-sim}, the proposed SAC-based method with a movable
RIS achieves the highest total reward and the fastest convergence. In
particular, SAC consistently outperforms TD3, and further improves upon DDPG,
achieving about $16\%$ higher reward at convergence. Compared with the SCA
benchmark, the performance gain increases to approximately $28\%$.
This improvement is attributed to SAC’s ability to jointly optimize RIS phases,
UAV location, and transmit power within a unified learning framework, whereas
SCA relies on sequential updates that may lead to suboptimal solutions.
The figure also includes two additional benchmarks with a fixed RIS:
(i) \textit{SAC (fix,opt)} with optimized RIS phases and (ii) \textit{SAC (fix,rand)} with random phases. The
results highlight the importance of RIS phase optimization, as random phases
significantly degrade performance due to the loss of control over the
cascaded BS--RIS--IR/UEHR links.

Fig.~\ref{fig3-sim} depicts SEE versus the transmit power
\(P_{\mathrm{max}}\) for the SCA, DDPG, and the proposed SAC-based schemes.
As expected, increasing \(P_{\mathrm{max}}\) improves SEE for all methods
due to higher achievable secrecy rates. Nevertheless, the SAC-based approach
consistently outperforms both SCA and DDPG across the entire power range.
At high transmit power levels, SCA becomes comparable to DDPG; however,
both remain inferior to SAC.
\begin{figure}[t]
    \centering
\includegraphics[width=\linewidth]{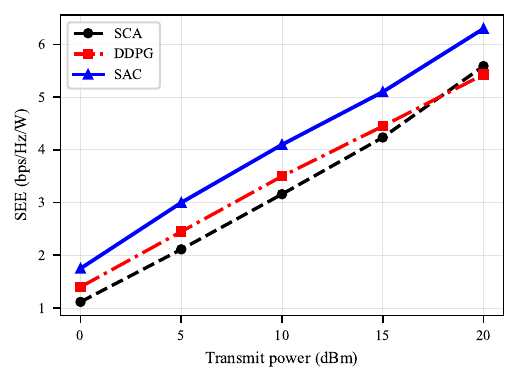}
    \caption{SEE versus transmit power for SCA, DDPG, and the proposed SAC-based scheme.}
    \label{fig3-sim}
\end{figure}
Fig.~\ref{fig:SEE_vs_M} illustrates the impact of the number of RIS elements $M$
on SEE under the ideal PCSI setting. SEE increases monotonically with $M$ for
all schemes due to the enhanced passive beamforming gain and stronger cascaded
links. The proposed SAC-based approach consistently achieves the highest SEE
across the entire range of $M$, highlighting its effectiveness in jointly
optimizing RIS phase shifts, UAV location, and transmit power. In contrast,
the SCA-based benchmark improves more rapidly at larger $M$ due to increased
optimization accuracy under PCSI, while DDPG saturates earlier, indicating
limited adaptability to the enlarged action space. Overall, the results confirm
the superior scalability and robustness of SAC with increasing RIS size.

\begin{figure}[t]
    \centering
\includegraphics[width=\linewidth]{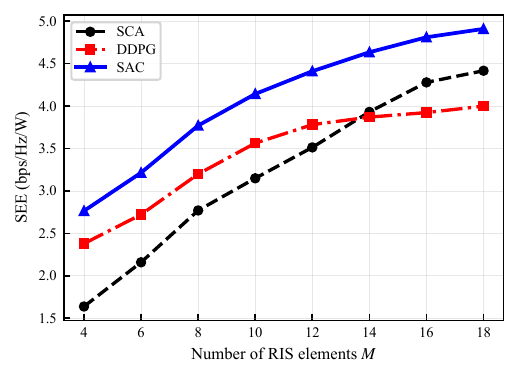}
    \caption{SEE versus the number of RIS elements $M$ at $P_{\mathrm{max}}=10$~dBm for different schemes.}
    \label{fig:SEE_vs_M}
\end{figure}


\subsection{Performance in the Non-Ideal Case}
We next evaluate a non-ideal setting with ICSI and discrete RIS
phases ($L=8$), where $E_h=0.02$ and $P_{\mathrm{max}}=10$~dBm.
\begin{figure}[t]

    \centering
    \includegraphics[width=\linewidth]{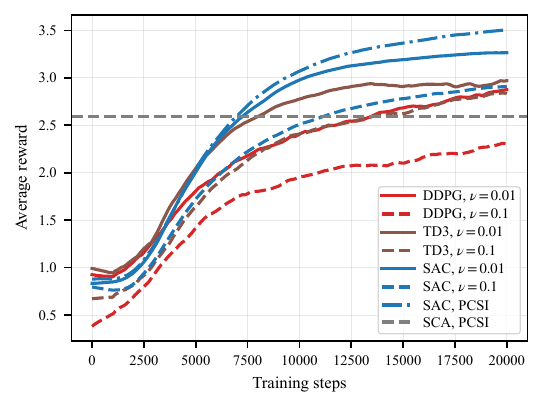}
    \caption{Convergence comparison under ICSI with discrete RIS phases
    ($L=8$) and $E_h=0.02$ at $P_{\mathrm{max}}=10$~dBm.}
    \label{fig2-sim}
\end{figure}
Fig.~\ref{fig2-sim} compares SAC, TD3, and DDPG under different CSI uncertainty radii
$\nu\in\{0.01,0.1\}$, with PCSI curves and the SCA baseline shown for reference.
As $\nu$ increases, the achievable reward decreases due to the more conservative
robust design. Nevertheless, SAC consistently outperforms both TD3 and DDPG and
exhibits faster convergence.
Notably, SAC with $\nu=0.1$ still surpasses DDPG with $\nu=0.01$,
demonstrating strong robustness under severe CSI uncertainty. This gain is attributed to the entropy-regularized policy and twin-critic structure of SAC, which improve exploration and reduce value overestimation under CSI uncertainty.

Fig.~\ref{fig_BS} compares the learning performance of SAC and DDPG under different mini-batch sizes $B$ with $\nu=0.01$, $P_{\max}=10$~dBm, and $E_h=0.01$. It is observed that SAC consistently achieves higher average rewards and faster convergence than DDPG across all batch sizes. Moreover, increasing $B$ results in smoother learning curves and improved convergence stability for both methods, as larger batches provide more accurate gradient estimates and reduce update variance. Nevertheless, even with small mini-batch sizes, SAC demonstrates superior robustness and stability, highlighting its stronger sample efficiency and suitability for robust UAV--RIS optimization.
\begin{figure}[t]
    \centering
\includegraphics[width=\linewidth]{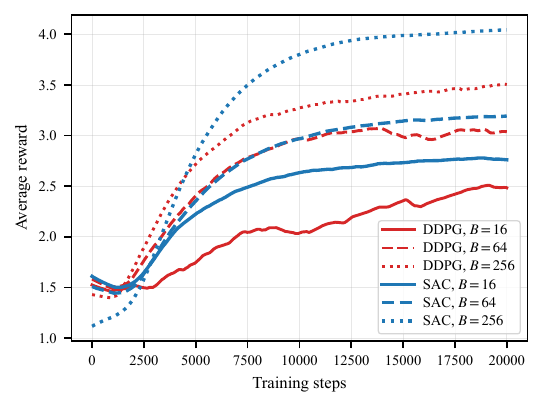}
    \caption{Convergence comparison of SAC and DDPG under different mini-batch sizes $B$ with $\nu=0.01$, $P_{\max}=10$~dBm, and $E_h=0.01$.}
    \label{fig_BS}
\end{figure}
%

To further evaluate the scalability of the proposed SAC framework, Fig.~\ref{fig_scalability} illustrates the WCSEE versus the number of legitimate users $K\in\{2,4,6,8,10\}$ for different numbers of UEHRs $J\in\{2,3,4,5\}$. The results are obtained under the ICSI setting with $\nu=0.01$, $E_h=0.01$, $N_t=10$, and $P_{\max}=30$~dBm.
\begin{figure}[t]
\centering
\includegraphics[width=\linewidth]{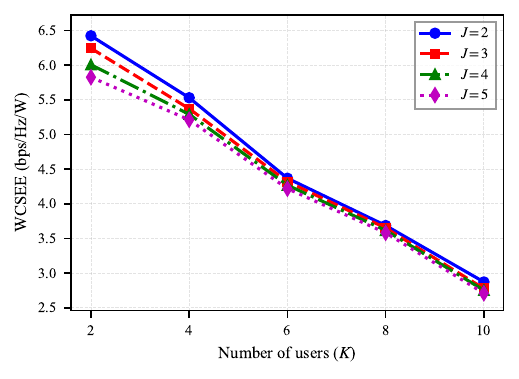}
\caption{WCSEE versus the number of legitimate users $K$ for different UEHR counts $J\in\{2,3,4,5\}$ with $\nu=0.01$, $E_h=0.01$, and $P_{\max}=30$~dBm.}
\label{fig_scalability}
\end{figure}
As shown in Fig.~\ref{fig_scalability}, the WCSEE decreases monotonically with increasing $K$ for all considered values of $J$. This degradation is mainly attributed to the limited transmit-power budget and the increasing spatial degrees-of-freedom (DoFs) required to suppress inter-user interference through ZF precoding. As the number of legitimate users increases, less power and spatial resources are available for improving the individual secrecy rates, resulting in lower overall WCSEE.

For a fixed $K$, increasing the number of UEHRs also reduces WCSEE. A larger $J$ introduces additional EH constraints and potential information-leakage channels, requiring more conservative resource allocation and RIS configuration to satisfy the worst-case secrecy and EH requirements. Notably, the performance gap among different $J$ values becomes progressively smaller as $K$ increases. For $K=2$, the WCSEE decreases from approximately $6.42$ to $5.85$~bps/Hz/W as $J$ increases from $2$ to $5$, indicating a more pronounced impact of additional UEHRs when sufficient spatial resources are available. In contrast, when $K=10$, all curves converge to approximately $2.8$~bps/Hz/W, suggesting that the performance becomes primarily constrained by the increasing multi-user interference-management burden rather than by the additional UEHRs. Overall, these results demonstrate that the proposed SAC framework maintains stable performance as the network becomes increasingly dense.
\section{Conclusions and Future Work}\label{conc}

This paper developed a robust framework for secure and energy-efficient resource allocation in UAV-mounted RIS-assisted multi-user MISO SWIPT networks with distributed UEHRs. The proposed formulation jointly considers the UAV hovering position, BS transmit power allocation, discrete RIS phase shifts, nonlinear EH, and norm-bounded CSI uncertainty, resulting in a challenging WCSEE maximization problem. A BCD--SCA framework was developed as a tractable benchmark for continuous RIS phase shifts under perfect CSI, while a customized SAC framework was proposed for the general robust setting with discrete phase control and imperfect CSI.
Simulation results demonstrated that the proposed SAC approach achieves up to $28\%$ and $16\%$ higher WCSEE than the SCA benchmark and DDPG, respectively, while providing faster convergence and robust performance under CSI uncertainty and different system configurations. The results also showed that SAC maintains stable learning across different mini-batch sizes. Furthermore, the scalability study revealed that the WCSEE decreases with increasing numbers of legitimate users and UEHRs, whereas the performance gap associated with different UEHR densities becomes less pronounced in dense-user regimes. These results highlight the effectiveness of SAC for jointly optimizing the highly coupled continuous and discrete variables in secure UAV--RIS-enabled SWIPT networks.
Despite these promising results, several practical aspects remain to be addressed. The current framework considers a hovering UAV and therefore does not explicitly capture dynamic trajectory, mobility, or propulsion-energy constraints. Moreover, the bounded CSI uncertainty model does not account for potentially time-varying or statistically correlated channel-estimation errors, and practical RIS hardware impairments such as mutual coupling are not considered.

Future work will extend the framework to dynamic 3D UAV trajectory optimization with mobility and propulsion-energy constraints, as well as cooperative multi-UAV/multi-RIS systems using multi-agent reinforcement learning. Other promising directions include active and STAR-RIS architectures, more realistic RIS hardware models, data-driven CSI uncertainty, and time-varying user mobility.


\appendices
\numberwithin{equation}{section}
\section{Proof of Proposition~\ref{prop:ris_zf_affine}}
\label{app:proof_prop_ris_zf}

We outline the derivations of \eqref{eq:ris_zf_affine_a}--\eqref{eq:ris_zf_affine_e} as follows.

The legitimate SINR constraint in \eqref{eq:aux_sinrs1} can be expressed (after introducing $\rho_k$) as
\begin{equation}\label{eq:A1_legit}
\rho_k \le
\frac{\big|\mathbf{t}_{kk}^H\mathbf{s}\big|^2}{
\sum_{\ell\neq k}\big|\mathbf{t}_{k\ell}^H\mathbf{s}\big|^2+\sigma^2},
\quad \forall k\in\mathcal{K}.
\end{equation}
Rearranging \eqref{eq:A1_legit} yields the DC form
\begin{equation}\label{eq:A2_legit_dc}
\mathcal{M}_k(\mathbf{s})-\mathcal{N}_k(\mathbf{s},\rho_k)\le 0,
\end{equation}
where
$\mathcal{M}_k(\mathbf{s})\triangleq \sum_{\ell\neq k}\big|\mathbf{t}_{k\ell}^H\mathbf{s}\big|^2+\sigma^2$
and
$\mathcal{N}_k(\mathbf{s},\rho_k)\triangleq \frac{\big|\mathbf{t}_{kk}^H\mathbf{s}\big|^2}{\rho_k}$.
The nonconvexity is due to $\mathcal{N}_k(\mathbf{s},\rho_k)$.
Thus, at iteration $t$, we replace $\mathcal{N}_k(\mathbf{s},\rho_k)$ by its
first-order Taylor approximation $\Psi^{(t)}(\mathbf{s},\rho_k;\mathbf{t}_{kk})$,
which directly gives \eqref{eq:ris_zf_affine_a}.

For UEHR $j$ decoding user $k$, introduce auxiliary variables $\rho_{E,j,k}$ and
$\xi_{j,k}$ such that
\begin{equation}\label{eq:A3_eve_aux}
\begin{cases}
\big|c_{j,k}+\mathbf{t}_{j,k}^H\mathbf{s}\big|^2 \le \xi_{j,k}\rho_{E,j,k},\\[0.5mm]
\xi_{j,k} \le \sum_{\ell\neq k}\big|c_{j,\ell}+\mathbf{t}_{j,\ell}^H\mathbf{s}\big|^2+\sigma^2.
\end{cases}
\end{equation}
The product $\xi_{j,k}\rho_{E,j,k}$ in \eqref{eq:A3_eve_aux} is bilinear and nonconvex;
we apply the standard affine lower bound $\varTheta^{(t)}(\xi_{j,k},\rho_{E,j,k})$,
which yields \eqref{eq:ris_zf_affine_b}.
Moreover, each quadratic term $\big|c_{j,\ell}+\mathbf{t}_{j,\ell}^H\mathbf{s}\big|^2$ in the second inequality
of \eqref{eq:A3_eve_aux} is convex; thus, its first-order Taylor expansion provides a global lower bound.
Accordingly, we construct a convex surrogate constraint via inner approximation, leading to \eqref{eq:ris_zf_affine_c}.

Constraint \eqref{eq:ris_zf_affine_d} follows by rewriting
$f_{E,k}\ge \log_2(1+\rho_{E,j,k})$ into
$2^{f_{E,k}} \ge 1+\rho_{E,j,k}$ and then applying the first-order Taylor
approximation $\Gamma^{(t)}(f_{E,k})$ of $2^{f_{E,k}}$.
Finally, \eqref{eq:ris_zf_affine_e} is the epigraph form of the secrecy constraint.
This completes the proof.


\section{Proof of Proposition~\ref{prop:eh_inv_y_lb}}\label{app:eh_inv_y_lb}
For fixed $\{c_{j,\ell},\tilde D_{j,\ell}\}$ and real $x_{jb}\ge 0$,
expanding \eqref{eq:eh_uav} yields the convex quadratic
\begin{equation}\label{eq:gj_expand_app}
P_j^{\mathrm{EH}}(x_{jb})=U_j+S_jx_{jb}^2+2T_jx_{jb},
\end{equation}
where $U_j$, $S_j$, and $T_j$ are defined in Proposition~\ref{prop:eh_inv_y_lb}.
Next, we lower bound the linear term $2T_jx_{jb}$.
\begin{Lemma}\label{lem:young_lb_app}
For any $x_{jb}\in\mathbb{R}$ and any $\varepsilon_j>0$,
\begin{equation}\label{eq:young_lb_app}
2T_j x_{jb} \ \ge\ -\varepsilon_j T_j^2 - \frac{1}{\varepsilon_j}x_{jb}^2 .
\end{equation}
\end{Lemma}
\begin{proof}
Let $a=\sqrt{\varepsilon_j}\,T_j$ and $b=x_{jb}/\sqrt{\varepsilon_j}$.
Since $(a+b)^2\ge 0$, we have $2ab\ge -a^2-b^2$, i.e.,
$2T_jx_{jb} \ge -\varepsilon_jT_j^2 - x_{jb}^2/\varepsilon_j$, which proves \eqref{eq:young_lb_app}.
\end{proof}
Applying Lemma~\ref{lem:young_lb_app} to \eqref{eq:gj_expand_app} yields, for any
$\varepsilon_j>0$,
\begin{equation}\label{eq:gj_lb_x2_app}
P_j^{\mathrm{EH}}\ge \big(U_j-\varepsilon_jT_j^2\big)+\Big(S_j-\frac{1}{\varepsilon_j}\Big)x_{jb}^2
= A_j(\varepsilon_j)+B_j(\varepsilon_j)x_{jb}^2.
\end{equation}
Choosing $\varepsilon_j>1/S_j$ (for $S_j>0$) ensures $B_j(\varepsilon_j)>0$.
Finally, from \eqref{eq:yjb_uav} and
$x_{jb}=\big(d_j^2(\mathbf q)d_b^2(\mathbf q)\big)^{-\alpha/4}$, we obtain
\begin{equation}\label{eq:x2_vs_inv_y_app}
x_{jb}^2=\big(d_j^2(\mathbf q)d_b^2(\mathbf q)\big)^{-\alpha/2}
\ge \frac{1}{y_{jb}}.
\end{equation}
Substituting \eqref{eq:x2_vs_inv_y_app} into \eqref{eq:gj_lb_x2_app} 
yields
\(
P_j^{\mathrm{EH}}\ge A_j(\varepsilon_j)+B_j(\varepsilon_j)\,1/y_{jb},
\)
which coincides with \eqref{eq:eh_lb_inv_y_main} and completes the proof.

\bibliographystyle{IEEEtran}
\bibliography{reference}

\end{document}